\documentclass[12pt,draftclsnofoot,onecolumn]{IEEEtran}

\usepackage{amsmath, mathrsfs, mathtools}
\usepackage{amsfonts}
\usepackage{amssymb}
\usepackage{color}
\usepackage{multirow}
\usepackage{stmaryrd}

\makeatletter
\newcommand{\doublewidetilde}[1]{{%
  \mathpalette\double@widetilde{#1}%
}}
\newcommand{\double@widetilde}[2]{%
  \sbox\z@{$\m@th#1\widetilde{#2}$}%
  \ht\z@=.9\ht\z@
  \widetilde{\box\z@}%
}
\makeatother

\usepackage{caption}
\usepackage{subcaption}
\usepackage{float}

\usepackage{amsfonts}
\usepackage{amssymb}
\usepackage{color}
\usepackage{multirow}
\usepackage{graphicx}
\usepackage{tcolorbox}

\graphicspath{{figures/}}

\newcommand{\myhash}{%
	{\settoheight{\dimen0}{C}\kern-.05em\, \resizebox{!}{\dimen0}{\raisebox{\depth}{\#}}}}

\usepackage[numbers,sort&compress]{natbib}

\newcommand{\Sigmah}{{\Sigmam}_{\hv}}

\DeclareMathAlphabet\mathbfcal{OMS}{cmsy}{b}{n}

\def\rect{{\ttr\tte\ttc\ttt}}

\def\wh{\widehat}
\usepackage{multicol}

\usepackage{algorithm}
\usepackage{algpseudocode}
\usepackage{pifont}
\usepackage{varwidth}


\usepackage{mathbbol}

\def\mindex#1{\index{#1}}

\def\sq{\hbox{\rlap{$\sqcap$}$\sqcup$}}
\def\qed{\ifmmode\sq\else{\unskip\nobreak\hfil
\penalty50\hskip1em\null\nobreak\hfil\sq
\parfillskip=0pt\finalhyphendemerits=0\endgraf}\fi\medskip}

\long\def\defbox#1{\framebox[.9\hsize][c]{\parbox{.85\hsize}{%
\parindent=0pt
\baselineskip=12pt plus .1pt      
\parskip=6pt plus 1.5pt minus 1pt 
 #1}}}

\long\def\beginbox#1\endbox{\subsection*{}%
\hbox{\hspace{.05\hsize}\defbox{\medskip#1\bigskip}}%
\subsection*{}}

\def\endbox{}

\def\diag{{\text{diag}}}

\newsavebox{\junk}
\savebox{\junk}[1.6mm]{\hbox{$|\!|\!|$}}

\def\det{{\mathop{\rm det}}}

\def\argmin{\mathop{\rm arg\, min}}

\def\argmax{\mathop{\rm arg\, max}}

\def\bC{{\mathbb C}}

\def\bR{{\mathbb R}}

\def\bfA{{\bf A}}

\def\bfI{{\bf I}}

\def\bfS{{\bf S}}

\def\bfU{{\bf U}}

\def\bfa{{\bf a}}

\def\bff{{\bf f}}

\def\bfh{{\bf h}}

\def\bfu{{\bf u}}

\def\bfy{{\bf y}}

\def\scrC{{\mathscr{C}}}

\def\ttc{{\mathtt c}}

\def\tte{{\mathtt e}}

\def\ttr{{\mathtt r}}

\def\ttt{{\mathtt t}}

\def\sfF{{\sf F}}

\def\bfmath#1{{\mathchoice{\mbox{\boldmath$#1$}}%
{\mbox{\boldmath$#1$}}%
{\mbox{\boldmath$\scriptstyle#1$}}%
{\mbox{\boldmath$\scriptscriptstyle#1$}}}}

\def\bfmY{\bfmath{Y}}

\def\bfmhhaY{\bfmath{\hhaY}} 
\def\bfmhhaY{\hbox to 0pt{$\widehat{\bfmY}$\hss}\widehat{\phantom{\raise 1.25pt\hbox{$\bfmY$}}}}

\def\til={{\widetilde =}}

\def\clC{{\cal C}}

\def\clG{{\cal G}}

\def\clN{{\cal N}}

 \def\FRAC#1#2#3{\genfrac{}{}{}{#1}{#2}{#3}}

\def\ddtp{{\mathchoice{\FRAC{1}{d^{\hbox to 2pt{\rm\tiny +\hss}}}{dt}}%
{\FRAC{1}{d^{\hbox to 2pt{\rm\tiny +\hss}}}{dt}}%
{\FRAC{3}{d^{\hbox to 2pt{\rm\tiny +\hss}}}{dt}}%
{\FRAC{3}{d^{\hbox to 2pt{\rm\tiny +\hss}}}{dt}}}}

\def\average#1,#2,{{1\over #2} \sum_{#1}^{#2}}

\def\eye(#1){{\bf(#1)}\quad}

\newtheorem{theorem}{{\bf Theorem}}

\newtheorem{remark}{{\bf Remark}}

\newtheorem{lemma}[theorem]{{\bf Lemma}}

\def\eq#1/{(\ref{e:#1})}

\newcommand{\beqn}[1]{\notes{#1}%
\begin{eqnarray} \elabel{#1}}

\newcommand{\eeqn}{\end{eqnarray} }

\newcommand{\beq}[1]{\notes{#1}%
\begin{equation}\elabel{#1}}

\newcommand{\eeq}{\end{equation}}

\def\bdes{\begin{description}}
\def\edes{\end{description}}

\newcounter{rmnum}

\newcounter{anum}

{\end{list}}

\def\ass(#1:#2){(#1\ref{#1:#2})}

\def\ritem#1{
\item[{\sf \ass(\current_model:#1)}]
}

\newenvironment{recall-ass}[1]{%
\begin{description}
\def\current_model{#1}}{
\end{description}
}

\usepackage{float}
\usepackage{afterpage}
\usepackage{balance}

\newcommand{\SigmaS}{\Sigmam_{\hv} (\uv)}

\def\cg{{\clC\clN}} 

\long\def\comment#1{}

\newfont{\bb}{msbm10 scaled 1100}
\newcommand{\CC}{\mbox{\bb C}}

\newcommand{\RR}{\mbox{\bb R}}

\newcommand{\EE}{\mbox{\bb E}}

\newcommand{\av}{{\bf a}}
\newcommand{\bv}{{\bf b}}
\newcommand{\cv}{{\bf c}}

\newcommand{\fv}{{\bf f}}

\newcommand{\hv}{{\bf h}}

\newcommand{\uv}{{\bf u}}

\newcommand{\xv}{{\bf x}}
\newcommand{\yv}{{\bf y}}
\newcommand{\zv}{{\bf z}}

\newcommand{\Am}{{\bf A}}

\newcommand{\Dm}{{\bf D}}

\newcommand{\Id}{{\bf I}}

\newcommand{\Sm}{{\bf S}}
\newcommand{\Tm}{{\bf T}}
\newcommand{\Um}{{\bf U}}
\newcommand{\Wm}{{\bf W}}

\newcommand{\Xm}{{\bf X}}
\newcommand{\Ym}{{\bf Y}}
\newcommand{\Zm}{{\bf Z}}

\newcommand{\Ac}{{\cal A}}

\newcommand{\Lc}{{\cal L}}

\newcommand{\thetav}{\hbox{\boldmath$\theta$}}

\newcommand{\sigmav}{\hbox{\boldmath$\sigma$}}

\newcommand{\Lambdam}{\hbox{\boldmath$\Lambda$}}
\newcommand{\Deltam}{\boldsymbol{\Delta}}
\newcommand{\Sigmam}{\hbox{\boldmath$\Sigma$}}

\renewcommand{\det}{{\hbox{det}}}
\newcommand{\trace}{{\hbox{tr}}}
\renewcommand{\arg}{{\hbox{arg}}}

\newcommand{\herm}{{\sf H}}
\newcommand{\trasp}{{\sf T}}
\newcommand{\transp}{{\sf T}}
\renewcommand{\vec}{{\rm vec}}

\title{Structured Channel Covariance Estimation from Limited Samples for Large Antenna Arrays}
\author{Tianyu Yang$^1$, Mahdi Barzegar Khalilsarai$^1$, Saeid Haghighatshoar$^2$,  \\and Giuseppe Caire$^1$  
\thanks{$^1$Communications and Information Theory Group (CommIT), Technische Universit\"{a}t Berlin, 10587 Berlin, Germany (e-mail: \{tianyu.yang, m.barzegarkhalilsarai, caire\}@tu-berlin.de).}
\thanks{$^2$Saeid Haghighatshoar was with the CommIT Group, Technische Universit\"{a}t Berlin, 10587 Berlin, Germany. He is now with Swiss Center for Electronics and Mirotechnology (CSEM), 2002 Neuch\~{a}tel, Switzerland (e-mail: saeid.haghighatshoar@csem.ch).}
\thanks{Part of this work has been presented in IEEE International Conference on Communications (ICC) 2020 \cite{khalilsarai2020structured}.}
}

\begin{document}

\maketitle

\def\ful{f_\text{ul}}
\def\fdl{f_\text{dl}}
\def\asfc{\scrC}
\def\asful{\scrC_\text{ul}}
\def\asfdl{\scrC_\text{dl}}

\vspace{-1cm} 

\begin{abstract}
In massive multiple-input multiple-output (MIMO) systems, the knowledge of the users' channel covariance matrix is crucial for minimum mean square error (MMSE) channel estimation in the uplink as well as it plays an important role in several multiuser beamforming schemes in the downlink. 
Due to the large number of base station antennas in massive MIMO, accurate covariance estimation is challenging especially in the case where the number of samples is limited and thus comparable to the channel vector dimension. As a result, the standard sample covariance estimator may yield 
a too large estimation error which in turn may yield significant system performance degradation with respect to the ideal channel covariance knowledge case. 
To address such problem, we propose a method based on a parametric representation of the channel angular scattering function. The proposed parametric representation includes a discrete specular component which is addressed using the well-known MUltiple SIgnal Classification (MUSIC) method, and a diffuse scattering component, which is modeled as the superposition of suitable dictionary functions. To obtain the representation parameters we propose two methods, where the first solves a non-negative least-squares problem and the second maximizes the likelihood function using expectation-maximization. 
Our simulation results show that the proposed methods outperform the state of the art with respect to various estimation quality metrics and different sample sizes.  
\end{abstract}


\section{Introduction} \label{sec:intro}

Massive multiple-input multiple-output (MIMO) communication system, where the number of base station (BS) antennas $M$ is much larger than the number of single antenna users, has been shown to achieve high spectral efficiency in wireless cellular networks and enjoys various system level benefit, such as energy efficiency, inter-cell interference reduction, and dramatic simplification of user scheduling (e.g., see \cite{boccardi2014five,Larsson-book}). In a large number of papers on the subject, the knowledge of the uplink (UL) and downlink (DL) channel covariance matrix, i.e., of the correlation structure of the channel antenna coefficients at the BS array, is assumed and used for a variety of purposes, such as minimum mean square error (MMSE) UL channel estimation and pilot decontamination \cite{chen2016pilot,yin2016robust,haghighatshoar2017decontamination}, efficient DL multiuser precoding/beamforming design, especially in the frequency division duplexing (FDD) case \cite{khalilsarai2018fdd, boroujerdi2018low,haghighatshoar2018low, haghighatshoar2017massive, adhikary2013joint}, statistical channel state information (CSI) based transmission design and statistical beamforming \cite{qiu2017downlink, you2020spectral,liu2020statistical}.    

Under the usual assumption of wide-sense stationary (WSS) uncorrelated scattering (US) \cite{raghavan2010sublinear,haghighatshoar2018low,haghighatshoar2016massive,haghighatshoar2017massive,khalilsarai2018fdd}\footnote{The channel models in 3GPP standard TR 38.901 Subclause 7.5 equation 7.5-22 \cite{3gpp38901} and TR 25.996 equation 5.4-1 \cite{3gpp25996} specify the US property. Section 3 of the QuaDRiGa documentation \cite{jaeckel2014quadriga} also implicitly includes this assumption.}, the channel vector evolves over time as a vector-valued WSS process and its spatial correlation is frequency-invariant over 
a frequency interval significantly larger than the signal bandwidth, although much smaller than the carrier frequency.
In particular, in an orthogonal frequency-division multiplexing (OFDM) system, the channel spatial covariance is independent of time (OFDM symbol index) and frequency (subcarrier index). In order to capture the actual WSS statistics, samples sufficiently spaced over time and frequency must be collected (e.g., one sample for every resource block in the UL slots over which a given user is active). On the other hand, the WSS model holds only {\em locally}, over time intervals where the propagation geometry (angle and distances of multipath components) does not change significantly. Such time interval, often referred to as ``geometry coherence time'', is several orders of magnitude larger than the coherence time of the small-scale channel coefficients. For a typical mobile urban environment, the channel geometry coherence time is of the order of seconds, while the small-scale fading coherence time is of the order of milliseconds (see \cite{va2016impact} and references therein).\footnote{See also discussion in \cite[Section I]{haghighatshoar2016massive}, and \cite{mahler2015propagation} for a recent study based on data gathered from a channel sounder in an urban environment with a receiver moving at vehicular speed.} Hence, the BS can collect a window of tens-to-hundreds of noisy channel snapshots from UL pilot symbols sent by any given user,  and use these samples to produce a ``local'' estimate of the corresponding user covariance matrix which remains valid during a channel geometry coherence time. Such estimation must be repeated, or updated, at a rate that depends on the propagation scenario and user-BS relative motion. This discussion points out that the number of samples $N$ available for covariance estimation is limited and often comparable or even less than the number of BS antennas $M$. Therefore, the accurate estimation of the high-dimensional $M \times M$ spatial channel covariance matrix (both in the UL and in the DL) from a limited number of noisy samples is generally a difficult task. 

The simplest way to estimate the channel covariance matrices is the sample covariance estimator. Such estimation is asymptotically unbiased and consistent and works well when $N \gg M$. Unfortunately, as already noticed above, this is typically not the case in massive MIMO. Hence, the goal of this paper is to devise new parametric estimators that outperform the sample covariance estimator, as well as the other state-of-the-art methods proposed in the literature. In addition, an attractive feature of the proposed parametric estimation is that it lends itself to the extrapolation of the estimated channel covariance matrix from the UL to the DL frequency band. As the channel samples are collected by the BS through pilot symbols sent by the users on the UL, the UL channel covariance can be directly estimated via our method.  However, as pointed out above, several schemes for DL multiuser precoding/beamforming in FDD systems make use of the user channel covariance matrix in the DL, which differs from the UL covariance since the frequency separation between the UL and the DL bands is large. The estimation of the DL covariance from UL channel samples has been considered in several works and it is generally another challenging task \cite{khalilsarai2018fdd,miretti2018fdd,khalilsarai2019uplink,song2020deep,gonzalez2020fdd,decurninge2015channel}. We shall see that our scheme is able to accurately estimate the DL channel covariance by extrapolating (over frequency) the estimated parametric model in the UL. 

\subsection{Related Work}

Covariance estimation from limited samples is a very well-known classical problem in statistics. In the ``large system regime'' when both the vector dimension $M$ and the number of samples $N$ grow large with fixed sampling ratio $N/M$ of samples per dimension, a vast literature focused on the asymptotic eigenvalue distribution of sample covariance matrices estimator under specific statistical assumptions (see, e.g., \cite{marchenko1967distribution, hachem2005empirical, couillet2011random} and the references therein). Covariance estimation has reemerged recently in many problems in machine learning, compressed sensing, biology, etc. (see, e.g., \cite{pourahmadi2013high, ravikumar2011high, chen2011robust, friedman2008sparse} from some recent results).  What makes these recent works different from the classical ones is the highly-structured nature of the covariance matrices in these applications. A key challenge in these new applications is to design efficient estimation algorithms able to take advantage of the underlying structure to recover the covariance matrix with a small sample size. 

As we explain in the sequel, MIMO covariance matrices are also a  class of highly-structured covariance matrices due to the particular spatial configuration of the BS antennas and wireless propagation scenario. So, it is very important to design efficient algorithms that are able to take advantage of this specific structure. In this paper, we shall compare our proposed method with alternative approaches which can be regarded as the state-of-the-art for the specific case of wireless massive MIMO channels. The competitor methods shall be briefly discussed when presenting such comparison results.

\subsection{Contributions}

The main contributions of this work are listed below:
\begin{enumerate}

    \item Previous works have considered either {\em diffuse scattering} (the received power from any given angle direction is infinitesimal \cite{miretti2018fdd}), or separable discrete components \cite{stoica2011spice,xie2018channel,park2018spatial}. Interestingly, both classes of state of the art methods incur significant degradation when the actual channel model does not conform to the assumptions. In contrast, we propose a method that handles the more realistic 
    {\em mixed scattering model} where discrete and diffuse scattering are simultaneously present. This corresponds to a so-called ``spiked model'' 
    in the language of asymptotic random matrix theory (e.g., see \cite{couillet2011random}), where the spikes are the discrete scattering components.
    It should also be noticed that our model is not just ``one or many possible alternatives''. In fact, our model encompasses any possible received power distribution over the angle domain, referred to in this work as the {\em angular scattering function} (ASF).
    
   \item The new idea of the proposed general model consists of approximating the ASF in terms of atoms of a dictionary. This dictionary is formed by some 
   Dirac delta functions placed at the angle-of-arrival (AoA) of the discrete specular components and by a family of functions obtained by shifts of a template density function over the angle domain, in order to approximate the diffuse component of the ASF. Notice that the proposed dictionary is partially ``on the gird'' (the shifts of the template density function at integer multiples of some chosen discretization interval) and partially ``off the grid'' (the Diracs placed at the unquantized AoAs). In this sense, the method differs substantially from standard {\em compressed sensing} schemes, and does not assume sparsity of the ASF in the angle domain. In fact, the diffuse scattering component is not sparse at all, but is dense on given angular spread intervals. 
   
   \item In order to estimate the AoA of the discrete components and their number (model order), we propose to use MUltiple SIgnal Classification (MUSIC) \cite{stoica2005spectral}. This method is particularly suited to the problem at hand thanks to the provable {\em asymptotic consistency} for the corresponding spiked model. Notice that this property has not been proven for other super-resolution methods based on convexity and atomic norm, which have also the disadvantage of being significantly more computationally complex (see e.g., \cite[Table III]{chen2020new}) and numerically unstable.
    
   \item After having the estimated AoAs of the spikes via MUSIC, we propose two estimators, namely a constrained Least-Squares estimator and Maximum-Likelihood (ML) estimator, to obtain the parametric model coefficients. In particular, the maximization of the likelihood function in the ML estimator is obtained  via the expectation maximization (EM) initialized by the constrained Least-Squares solution. 
    
    \item We present several numerical results based on realistic channel emulator QUAsi Deterministic RadIo channel GenerAtor (QuaDriGa) \cite{jaeckel2014quadriga} with comparison with other competitor schemes. Our comparison shows that the proposed algorithms outperform 
    the competitor schemes over all benchmarks in a wide range of sampling ratio $N/M$. It should be noticed that the QuaDriGa channel generator is 
    agnostic and unaware of the proposed ASF representation and generates channel snapshots on the basis of a mixed physical/statistical model. Therefore, our results are not ``biased'' by generating channels that are matched to the assumed model. On the contrary, our results show that the proposed approach  provides improvements on the state of the art even if the underlying channel statistics do not reflect exactly the assumed model.
\end{enumerate}

\subsection{Used Notations}
An identity matrix with $K$ columns is denoted as ${\mathbf{I}}_K$.  An all-zero matrix with size $m \times n$ is denoted as $\mathbf{0}_{m \times n}$.  $\text{diag}(\cdot)$ returns a diagonal matrix.  $\|\cdot\|_{\sfF}$ returns Frobenius norm of a matrix. $\delta(\cdot)$ denotes the Dirac delta function.  We use $[n]$ to denote the ordered integer set $\{1, 2, \ldots, n\}$. 

\section{System Model}\label{sec:model}
\begin{figure}[t]
    \centering
    \includegraphics[width=0.5\textwidth]{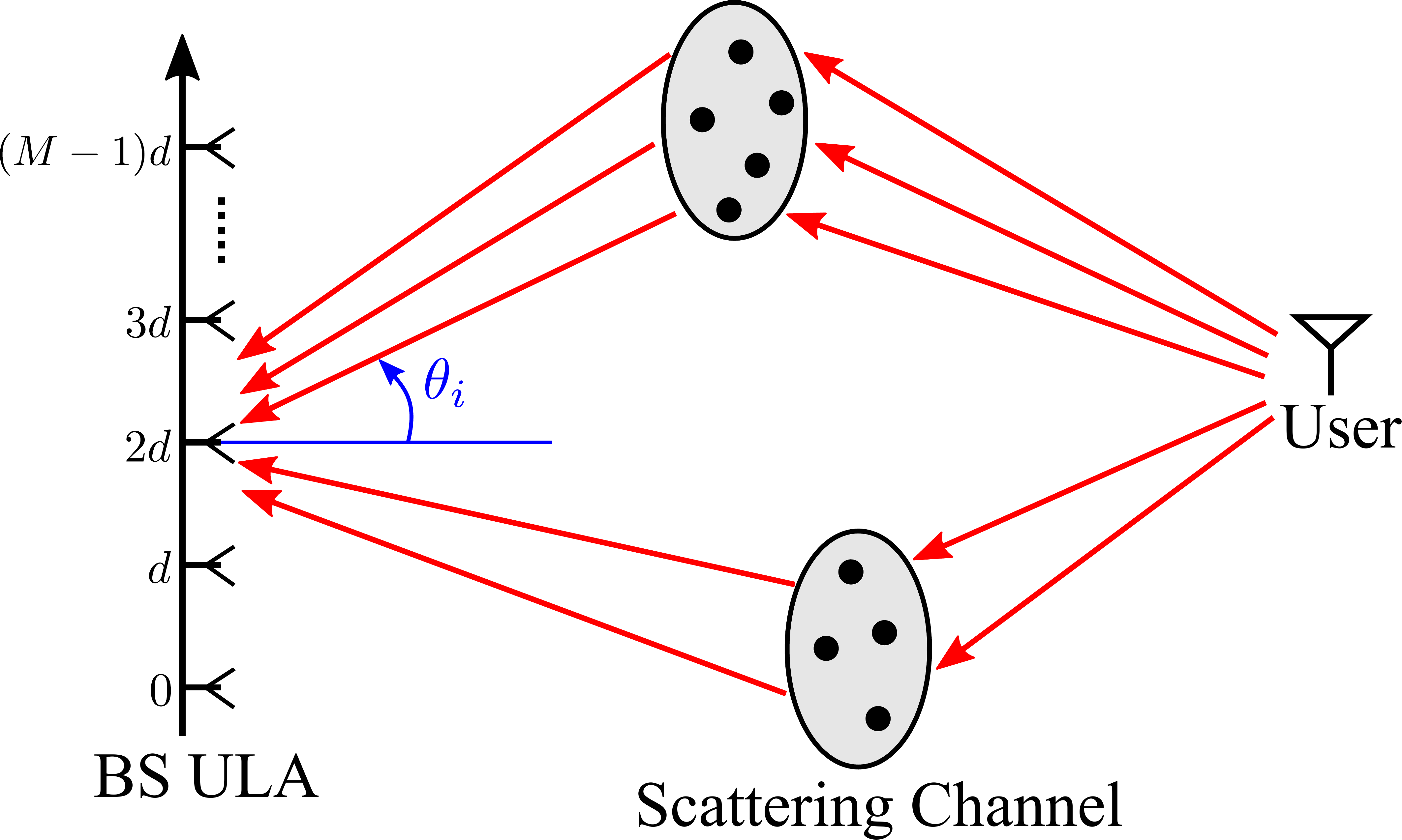}
    \caption{A cartoonish representation of a multipath propagation channel, where the user signal is received at the BS through two scattering clusters.}
    \label{fig:systemmodel}
\end{figure}

We consider a typical single-cell massive MIMO communication system, where a BS equipped with a uniform linear array (ULA) of $M$ antennas  communicates with multiple users through a scattering channel.\footnote{Notice that in the case of a multi-cell system with pilot contamination the same techniques of this paper can estimate the covariance of the sum of multiple channels, and a pilot decontamination scheme as described in \cite{haghighatshoar2017decontamination} can be applied. However, this goes beyond the scope of this paper.} Fig. \ref{fig:systemmodel} visualizes the propagation model based on multipath clusters, which is physically motivated and widely adopted in standard channel simulation tools such as QuaDriGa \cite{jaeckel2014quadriga}. During UL transmission, on each time-frequency resource block (RB) $s$ the BS receives an UL user pilot carrying a measurement for the channel vector $\hv[s]$. We assume that the window of $N$ samples collected for covariance estimation is designed such that the samples are enough spaced in the time-frequency domain and resulting in statistically independent channel snapshots $\{\hv[s]: s = [N]\}$. Meanwhile, the whole window spans a time significantly shorter than the geometry coherence time so that the WSS assumption holds (see discussion in Section \ref{sec:intro}) and the channel snapshots are identically distributed.  The channel vectors are given by \cite{sayeed2002deconstructing}
\begin{equation}\label{eq:ch_expression}
\hv [s] = \int_{-1}^1 \rho (\xi;s) \bfa(\xi) d \xi,~s\in [N],
\end{equation}
where $\rho(\xi;s)$ is the channel complex coefficient at the normalized 
AoA $\xi = \frac{\sin(\theta)}{\sin(\theta_{\text{max}})} \in [-1,1)$, where $\theta_{\text{max}} \in [0,\frac{\pi}{2}]$ is the maximum array angular aperture. $\bfa(\xi) \in \bC^M$ denotes the array response vector at $M$ BS antennas over $\xi$, whose $m$-th element is given as $[\av(\xi)]_m=e^{j\frac{2\pi d}{\lambda_0}m \xi\sin(\theta_{\text{max}})}$, where $d$ denotes the antenna spacing and $\lambda_0$ denotes the carrier wavelength. For convenience we assume the antenna spacing to be $d = \frac{\lambda_0}{2\sin(\theta_{\text{max}})}$. Thus, the array response vector is given as
\begin{equation}\label{eq:a_vec}
    \av (\xi)  = \left[1, e^{j\pi\xi},\dots, e^{j\pi(M-1)\xi}\right]^{\trasp}.
\end{equation}
The channel coefficient $\rho (\xi;s)$ models the small-scale multipath fading component at a given AoA and it is 
modeled as a complex circularly symmetric Gaussian process with respect to the continuous variable $\xi$. 
Due to the WSS property, the channel second-order statistics are invariant with respect to the index $s \in [N]$. 
In particular, $\rho (\xi;s)$ has mean zero and variance 
$\EE\left[\rho \left(\xi;s\right)\rho^\ast \left(\xi;s\right)\right] =  \gamma\left(\xi\right)$. 
The function $\gamma : [-1,1]\to  \mathbb{R}_+$ is a real non-negative measure that describes how the channel energy is distributed across the angle domain. 
We refer to $\gamma(\xi)$ as the channel ASF.  As a result, the channel spatial covariance matrix, describing the correlation of the channel coefficients at the different antenna elements, is given by  
\begin{align}\label{eq:cov_mat}
\Sigmam_\bfh=\EE \left[\bfh[s]\bfh[s]^\herm\right]=\int_{-1}^1 \gamma(\xi) \bfa(\xi) \bfa(\xi)^\herm d \xi.
\end{align} 
Notice that $\Sigmam_\bfh$ is Toeplitz. This fact is verified when all the scattering clusters (see Fig.~\ref{fig:systemmodel}) are in the far field of the BS array. 
As a note of caution, we hasten to say that the model must be reconsidered in the case of ``extra-large aperture'' arrays (e.g., see \cite{de2020non}), where the Toeplitz form does not hold any longer.
At RB $s$, the received pilot signal at the BS is given as 
\begin{equation}
    \yv[s] = \hv[s] x[s] + \zv[s],~s\in[N],   
\end{equation}
where $x[s]$ is the pilot symbol and $\zv[n] \sim \cg (\mathbf{0},N_0 \mathbf{I}_M)$ is the additive white Gaussian noise (AWGN). Without loss of generality, we assume that the pilot symbols are normalized as $x[n]=1, \; \forall\, s \in [N]$. The goal of this work is to estimate the channel covariance matrix $\Sigmam_\bfh$ with the given set of $N$ noisy channel observations $\{ \yv[s]: s\in [N]\}$. 

\subsection{Sample Covariance Matrix}

We start by reviewing the sample covariance estimator. For known noise power $N_0$ at the BS, the sample covariance matrix is given by
\begin{equation}\label{eq:sample_cov}
\widehat{\Sigmam}_\hv = \frac{1}{N} \sum_{s=1}^{N} \yv[s] \yv[s]^\herm - N_0 \mathbf{I}_M.
\end{equation}
This is a consistent estimator, in the sense that it converges to the true covariance matrix as $N \rightarrow \infty$ \cite[Section 1.2.2]{wainwright2019high}. 
The mean square (Frobenius-norm) error incurred by the sample covariance estimator is given as \cite{wainwright2019high}
$ \EE \left[\left\|\widehat{\Sigmam}_\hv-\Sigmam_\hv\right\|^2_{\sfF}\right] = \frac{\trace\left(\Sigmam_\bfh\right)^2}{N}$. 
By applying the Cauchy-Schwarz inequality to the singular values of $\Sigmam_\hv$, it is seen that $\trace (\Sigmam_\hv) \leq \|\Sigmam_\hv\|_{\sfF}\sqrt{\text{rank}(\Sigmam_\hv)}$, which together with the estimation error expression yields the upper bound to the normalized mean squared error $\EE \left[\frac{\|\widehat{\Sigmam}_\hv-\Sigmam_\hv\|^2_{\sfF}}{\|\Sigmam_\hv\|^2_{\sfF}}\right] \leq \frac{\text{rank}(\Sigmam_\bfh)}{N}$.
As already discussed, a relevant and interesting regime for massive MIMO is when $N$ and $M$ are of the same order. 
From the above analysis, it is clear that the sample covariance estimator yields a small error if $\text{rank}(\Sigmam_\bfh) \ll N$. For example, if the scattering contains only a finite number of discrete components (e.g., as assumed in \cite{stoica2011spice,park2018spatial,xie2018channel}), $\text{rank}(\Sigmam_\bfh)$ is small even if $M$ is very large. 
In contrast, if $\gamma(\xi)$ contains a diffuse scattering component, i.e., if its cumulative distribution function $\Gamma(\xi) = \int_{-1}^{\xi} \gamma(\nu) d\nu$ is piecewise continuous with strictly monotonically increasing segments,  then $\text{rank}(\Sigmam_\bfh)$ increases linearly with $M$ 
(see \cite{adhikary2013joint}) and  the error incurred by the sample covariance estimator can be large. On the other hand, the presence of discrete scattering components implies that $\gamma(\xi)$ contains Dirac delta functions (spikes) and therefore it is not squared-integrable. 
This poses significant problems for estimation methods that assume $\gamma(\xi)$ to be an element in a Hilbert space of functions (e.g., the method proposed in \cite{miretti2018fdd}). The challenge tackled in this work is to handle both the small sample regime $N/M \leq 1$ and the presence 
of both discrete and diffuse scattering. 

\subsection{Structure of the Channel Covariance Matrix}

As said, the ASF describes how the received signal power is distributed over the AoA domain. 
The signal from the UE to the BS array propagates through a given scattering environment. The line-of-sight (LoS) path (if present), 
specular reflections, and wedge diffraction occupy extremely narrow angular intervals.
This is usually modeled in a large number of papers as the superposition of discrete separable angular components coming at normalized AoAs $\{\phi_i\}$. In particular, it is assumed that the general form (\ref{eq:ch_expression}) reduces to the discrete sum of $r$ paths $\hv[s] = \sum_{i=1}^r \rho_i[s] \av(\phi_i)$, with corresponding ASF 
$\gamma(\xi) = \sum_{i=1}^r c_i \delta(\xi - \phi_i)$ and covariance matrix $\Sigmam_{\hv} = \sum_{i=1}^r c_i \av(\phi_i) \av(\phi_i)^\herm$, where $c_i = \EE[|\rho_i[s]|^2]$.  However, it is well-known from channel sounding observations (e.g., see \cite{hanssens2018modeling}) 
and widely treated theoretically (e.g., see \cite{sayeed2002deconstructing}) that {\em diffuse scattering} is typically also present,  
and may carry a very significant part of the received signal power especially at frequencies below 6GHz. 
In this case, scattering clusters span continuous intervals over the AoA domain. 
In order to encompass full generality,  we model the ASF  $\gamma(\xi)$ as a {\em mixed-type} distribution \cite[Section 5.3]{gubner2006probability}
including discrete and diffuse scattering components:  
\begin{equation}\label{eq:gamma_decomp}
    \gamma (\xi) = \gamma_d (\xi) + \gamma_c (\xi) = \sum_{i=1}^{r} c_i \delta (\xi - \phi_i) \, +\, \gamma_c (\xi), 
\end{equation}
where $\gamma_d(\xi)$ models the power received from $r \ll M$ discrete paths and $\gamma_c(\xi)$ models the power coming from diffuse scattering clusters. 
Since the ASF can be seen as a (generalized) density function, we borrow the language of discrete and continuous random variables and refer to  $\gamma_d(\xi)$ and to $\gamma_c(\xi)$ as the {\em discrete} and the {\em continuous} parts of the ASF, respectively.\footnote{As for probability density functions,  $\gamma_c(\xi)$ needs not be continuous, but its cumulative distribution function $\Gamma_c(\xi) = \int_{-1}^\xi \gamma_c(\nu) d\nu$ is a continuous function.} 
Plugging \eqref{eq:gamma_decomp} into \eqref{eq:cov_mat} we obtain a corresponding  decomposition of the channel covariance matrix as
\begin{equation}\label{eq:channel_cov}
\begin{aligned}
\Sigmah =  \Sigmah^d + \Sigmah^c &=  \sum_{i=1}^{r} c_i \av (\phi_i) \av (\phi_i)^\herm
 + \int_{-1}^1 \gamma_c (\xi) \av (\xi) \av (\xi)^\herm d \xi.
\end{aligned}
\end{equation}
Notice that, in a typical massive MIMO scenario, $\text{rank}(\Sigmah^d) = r$ is much smaller than $M$. 


\section{Dictionary-based Parametric Representation and Covariance Estimation} \label{sec:channel}
An outline of the steps taken by our proposed method is given in the following:

i)  \textit{Spike Location Estimation for $\gamma_d(\xi)$}: We apply the MUSIC algorithm \cite{stoica2005spectral} to estimate the AoAs of the spike components, i.e., the angles $\{ \phi_i \}_{i=1}^r$ in \eqref{eq:gamma_decomp}, from the $N$ noisy samples $\{\yv [s] : s\in [N]\}$. We let $\{\widehat{\phi}_i \}_{i=1}^{\widehat{r}}$ denote the estimated AoAs, where also the number of spikes $\widehat{r}$ is estimated (not assumed known). 
This is detailed in Section \ref{sec:MUSIC}.  For complete estimation of the discrete part $\Sigmah^d$, the weights $\{c_i\}_{i=1}^{\widehat{r}}$ need to be further recovered. This is done jointly with the coefficients of the continuous part,  as elaborated in the next step.

ii) \textit{Dictionary-based Method for Joint Estimation of $\gamma_d(\xi)$ and $\gamma_c(\xi)$}: 
We assume that the ASF continuous part can be written as
\begin{equation} \label{dicdic}
    \gamma_c(\xi) = \sum_{i=1}^G b_i \psi_i(\xi), 
\end{equation}
where $\clG_c=\{\psi_i(\xi): i \in [G]\}$ is a suitable dictionary of non-negative density functions (not containing spikes). Now the goal is to estimate the model parameters  $\{c_i \}_{i=1}^{\widehat{r}}$ and $\{b_i\}_{i=1}^G$ assuming the form of ASF $\gamma(\xi) = \sum_{i=1}^{\widehat{r}} c_i \delta(\xi - \widehat{\phi}_i) + \sum_{i=1}^G b_i \psi_i(\xi)$. The parameters estimation can be formulated as a constrained Least-Squares problem, as detailed in Section \ref{sec:LS}. In particular, if the functions in $\clG_c$ have disjoint support, we obtain a Non-Negative Least-Squares (NNLS) problem, while if the functions in $\clG_c$ have overlapping supports, we obtain a Quadratic Programming (QP) problem. As an alternative, we can use ML estimation. The log-likelihood function is a difference of concave functions in the model parameters, and can be maximized using majorization-minimization (MM) methods. However, general MM approaches are prohibitively computationally complex for typical values of $M$ arising in massive MIMO.  It turns out that  when $\clG_c$ is also formed by Dirac deltas 
spaced on a discrete grid, the likelihood function maximization can be obtained through EM with much lower complexity. Since in general EM is guaranteed to converge to local optima, the initialization plays an important role. We propose to use the result of the (low-complexity) NNLS estimator as initial point for the EM iteration. The resulting method is detailed in Section \ref{sec:ML}. 
	
iii) \textit{From ASF to Covariance Estimation}: Finally, having estimated $\gamma_d(\xi)$ and $\gamma_c(\xi)$, we estimate the covariance $\Sigmah$ via \eqref{eq:channel_cov}. In particular, since $\gamma(\xi)$ depends only on the scattering geometry and it is invariant with frequency,\footnote{This statement holds over not too large frequency ranges, where the scattering properties of materials is virtually frequency independent. For example, this property holds for the UL and DL carrier frequencies of the same FDD system, however it does not generally hold over much larger frequency ranges. For example, the ASFs of the same environment at (say) 3.5 GHz and at 28 GHz are definitely different, although quite related \cite{ali2016estimating}.} the mapping $\gamma(\xi) \rightarrow \Sigmam_\hv$ defined by \eqref{eq:channel_cov} can be applied for different carrier frequencies by changing the wavelength parameter $\lambda_0$ in the expression of the array response vector $\av(\xi)$. Specifically, replacing $\av(\xi)$ in \eqref{eq:channel_cov}  with  $\av(\nu\xi)$ where $\nu$ is a wavelength expansion/contraction coefficient, we obtain an estimator at the carrier frequency $f_\nu = \nu f_0$, where $f_0 = c_0/\lambda_0$ is the UL carrier frequency and $c_0$ denotes the speed of light. The resulting covariance estimator is given by 
\begin{align}
    \widehat{\Sigmam}_{\hv}^{(\nu)} = \sum^{G + \widehat{r}}_{i=1} u_i^\star \Sm_i^{(\nu)},
\end{align}
where $\{u_i^\star : i \in [G + \widehat{r}]\}$ are the estimated model parameters, and where  we define
$\Sm^{(\nu)}_i = \int^1_{-1}\psi_i(\xi) \av(\nu \xi)\av^\herm(\nu\xi) d\xi$ for $i \in [G]$ and $\Sm^{(\nu)}_{G+i} = \av(\nu \widehat{\phi}_{i}) \av^\herm(\nu \widehat{\phi}_i)$ for $i \in [\widehat{r}]$. In particular, for the DL carrier we have $\nu  > 1$ since in typical cellular systems the DL carrier frequency is higher than the UL carrier frequency. 
	
\subsection{Discrete ASF Support Estimation}\label{sec:MUSIC}

In this part, we first estimate the model order $r$ by applying the well-known minimum description length (MDL) principle \cite{schwarz1978estimating}. Then, we use the MUSIC method (e.g., see \cite{schmidt1986multiple,stoica1989music} and references therein), to estimate the locations $\{\phi_i\}_{i=1}^r$ of the spikes in discrete part of ASF $\gamma_d(\xi)$.
Given the noisy samples $\Ym = \{y[s]\}_{s=1}^N$, the sample covariance of $\Ym$ is given as
\begin{align}\label{eq:sample_cov_y}
    \widehat{\Sigmam}_\bfy = \frac{1}{N} \sum_{s=1}^{N} \bfy[s] \bfy[s]^\herm.
\end{align}
Let $\widehat{\Sigmam}_\bfy = \widehat{\bfU} \widehat{\Lambdam} \widehat{\bfU}^\herm$ be the eigendecomposition of $\widehat{\Sigmam}_\bfy$, where $\widehat{\Lambdam}=\diag(\widehat{\lambda}_{1,M}, \dots, \widehat{\lambda}_{M,M})$ denotes the diagonal matrix consisting of the eigenvalues of $\widehat{\Sigmam}_\bfy$. Without loss of generality, we assume that the eigenvalues are ordered as $\widehat{\lambda}_{1,M}\geq \dots \geq  \widehat{\lambda}_{M,M}$. 

First, we wish to estimate the model order $r$. We adopt the classical MDL method provided in \cite{wax1985detection}, which was designed for estimating the number of sources impinging on a passive array of sensors under white Gaussian noise with unknown noise power. Specifically, assuming that the covariance matrix of the observation has the form 
\begin{equation}\label{eq:sigmay_MDL}
    \widetilde{{\Sigmam}}_{\yv}=\Sigmah^d + \sigma^2\mathbf{I}_M,
\end{equation}
where $\Sigmah^d$ is the rank-$r$ covariance matrix in \eqref{eq:channel_cov} and $\sigma^2$ is an unknown parameter.\footnote{Note that $\widetilde{{\Sigmam}}_{\yv} \neq \mathbb{E}[\yv[s]\yv[s]^\herm]$ is not the true covariance of samples, since the contribution of diffuse clusters is not considered. 
Nevertheless, in Remark~\ref{remark_MDL} we argue that the MDL method in \cite{wax1985detection} can be applied to our problem without significant degradation by simply ignoring the presence of the diffuse component.}  
The MDL approach selects a model order $k \in\{0,1,\dots,M-1\}$ from a parameterized  family of probability densities $f(\Ym|\thetav^{(k)})$ that minimize the 
so-called total description length of the samples, which is tightly approximated by Rissanen bound \cite{barron1998minimum} (after neglecting $o(1)$ terms) as
\begin{equation}
    \text{MDL} = -\log (f(\Ym|\thetav^{(k)})) + \frac{1}{2}|\thetav^{(k)}|\log (N),
\end{equation}
where $|\thetav^{(k)}|$ is the number of free parameters in $\thetav^{(k)}$. The number of spikes is estimated by minimizing the MDL metric as \cite{wax1985detection} 
\begin{align}  
    \widehat{r} &= \arg\min_k \; \left \{ L(\thetav^{(k)}) + \frac{1}{2}k(2M-k)\log(N) \right \} ,   \label{bubu}
\end{align}
where $L(\thetav^{(k)}) = -\log f(\Ym|\thetav^{(k)})$ is the minus log-likelihood of $\Ym$, given by $L(k) = N(M-k)\log\left(\frac{a(k)}{b(k)}\right)$ with $a(k) = \frac{1}{M-k}\sum^M_{i=k+1}\widehat{\lambda}_{i,M}$ and when $k=0$, $b(k) = 1$, when $k > 0$, $b(k) = (\prod^k_{i=1}\widehat{\lambda}_{i,M})^{-\frac{1}{M-k}}$.\footnote{Note that our expression of $b(k)$ is different from that in \cite{wax1985detection}. However, it can be easily shown that the resulting 
MDL expressions differ only by a constant $M$, which has no influence in the minimization in \eqref{bubu}. Here 
$b(k)$ has been modified to avoid the product of small (nearly zero when sample size is quite small) eigenvalues of $\widehat{\Sigmam}_\bfy$.}   

\begin{remark}\label{remark_MDL} 
It is noticed that the proposed MDL approach does not explicitly consider the diffuse part. It considers only a spike signal subspace with dimension $r$ corresponding to the $r$ largest eigenvalues of $\widetilde{{\Sigmam}}_{\yv}$ and a white noise with power $\sigma^2$ resulting in $M-r$ smallest eigenvalues of $\widetilde{{\Sigmam}}_{\yv}$ equaling to $\sigma^2$. The diffuse component in our model can be considered as an additional {\em spatially colored noise} 
as far as the spike estimation is concerned. 
It is shown in \cite{xu1995analysis} that in the presence of colored Gaussian noise, MDL tends to overestimate the model order with increasing number of samples $N$. However, we shall see later that overcounting the spikes is not catastrophic in the overall covariance estimation scheme, since the coefficients of fictitious spikes are typically estimated as near zero in the model coefficient estimation step.  In other words, it is always better to overestimate the number of spikes than to underestimate them so that the true spikes are not missed. In \cite{liavas2001behavior} it is shown that undermodeling may happen when the signal and the noise eigenvalues are not well separated and the noise eigenvalues are clustered sufficiently closely. 
We show that the undermodeling is not likely to happen in our case by showing that the gap between $r$ large (containing the contribution of the spikes) and $M-r$ small eigenvalues (containing only the contribution of the diffuse scattering and the noise) of the sample covariance  $\widehat{\Sigmam}_\bfy$ becomes larger and larger as the number of antennas $M$ increases. 
This can be explained by noticing that, as the array angular resolution increases, the amount of received signal energy in each angular bin decreases with the bin
width in the bins that contain no spikes, while it remains roughly constant with $M$  if the angular bin contains a spike.  More precisely, an angular bin of width $2/M$ and centered at $\xi$ contains a received signal power proportional to $2 (\gamma_c(\xi) + N_0)/M$ if no spike falls in the interval, and to $c_i + 2 (\gamma_c(\xi) + N_0)/M$ if the spike located at $\phi_i$ falls in the interval. By Szeg\"o's theorem (e.g., see \cite{adhikary2013joint} and references therein),  the eigenvalues of the covariance matrix $\Sigmam_\yv = \Sigmam_\hv + N_0 \Id_M$ converge asymptotically as $M \rightarrow \infty$ to the energy received on equally spaced angular bins of width $2/M$ over the interval $[-1, 1]$ in the $\xi$ domain. Fig.~\ref{fig:eig_vec} corroborates this showing the separation between the eigenvalues of the sample covariance matrix $\widehat{\Sigmam}_\bfy$ for different number of antennas $M=25,~50,~100$ with fixed sample size to channel dimension ratio $N/M=2$ and the same channel geometry defined by the ASF
\begin{equation}\label{eq:ex_asf}
\gamma (\xi) = \rect_{[-0.7,-0.4]} +\rect_{[0,0.6]} + (\delta (\xi+0.2)+\delta (\xi-0.4))/2,
\end{equation}
where $\rect_{\Ac} $ is 1 over the interval $\Ac$ and zero elsewhere. This ASF contains $r=2$ spikes and two rectangular-shaped diffuse components. The SNR is set to $20$ dB. For a large enough number of antennas (and even for a moderate number such as $M=25$) the eigenvalue distribution shows a significant jump, such that the two largest eigenvalues ``escape" from the rest. Note that, by increasing the number of antennas, this separation becomes more and more significant. Hence, the undermodeling of MDL in the relevant case of massive MIMO is unlikely to occur.  \hfill $\lozenge$
\end{remark}

\begin{figure*}[t]
		\centering
		\begin{subfigure}[b]{0.31\textwidth}
			\includegraphics[width=\textwidth]{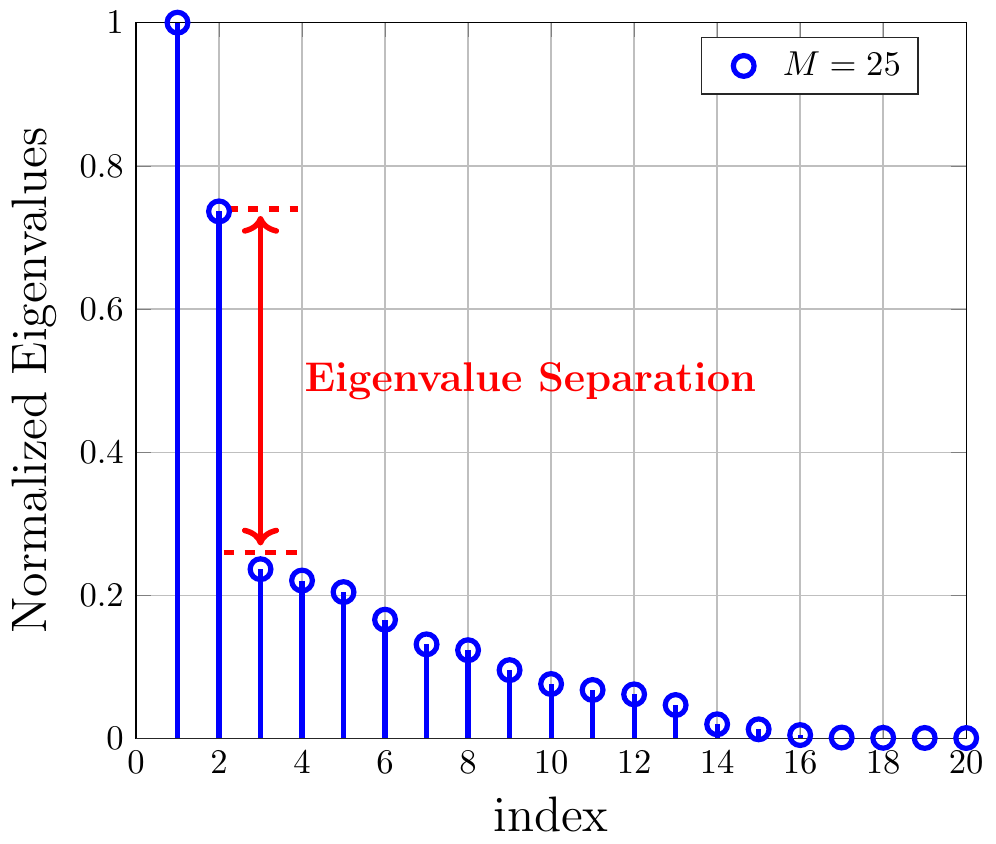}
		\end{subfigure}
		~ 
		\begin{subfigure}[b]{0.31\textwidth}
			\includegraphics[width=\textwidth]{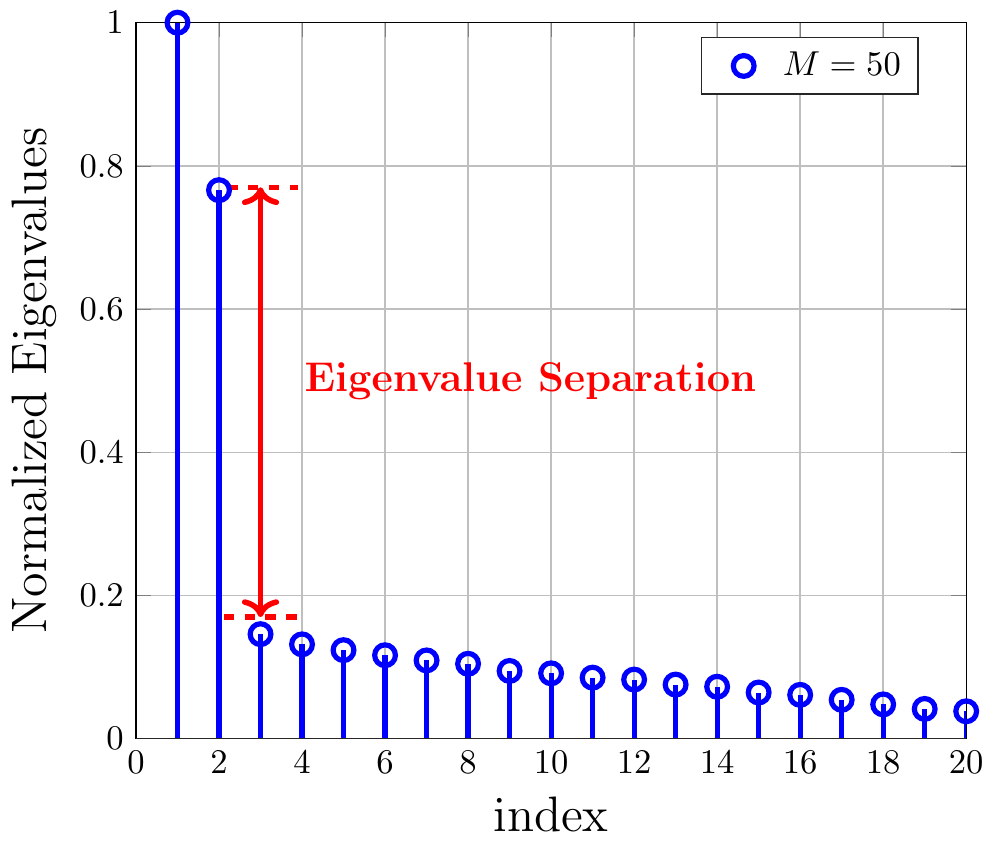}
		\end{subfigure}
		~
		\begin{subfigure}[b]{0.31\textwidth}
		\includegraphics[width=\textwidth]{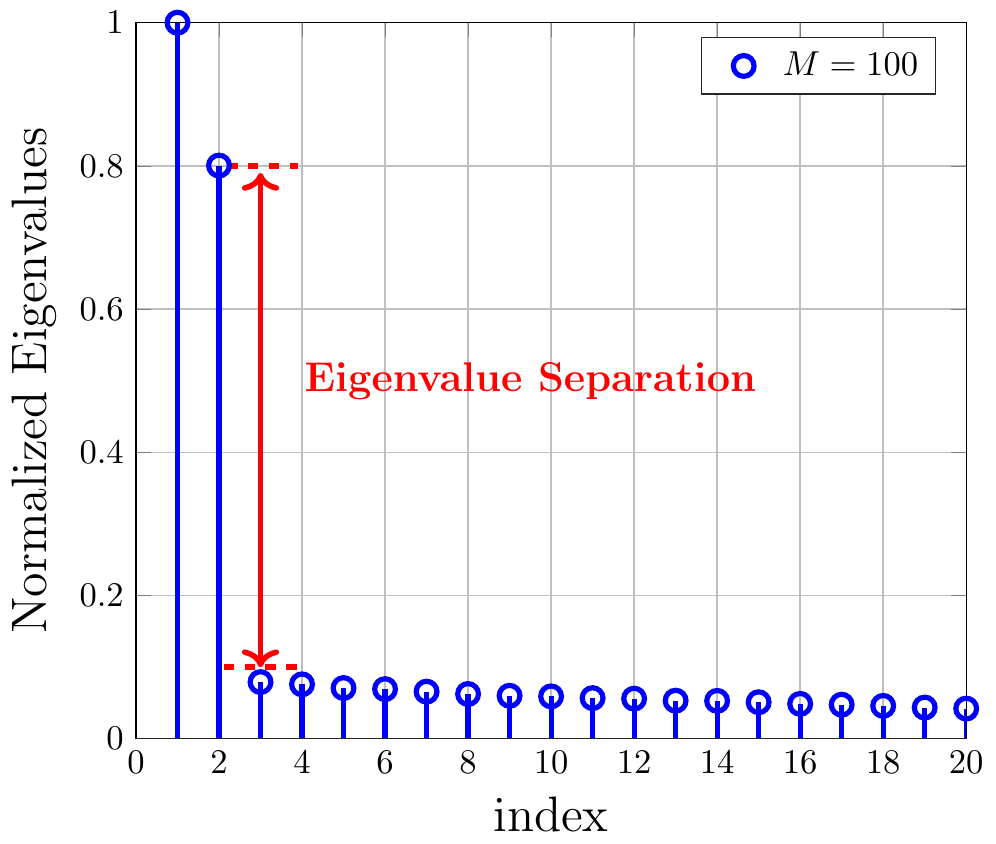}
\end{subfigure}
\caption{ Eigenvalue distribution for the sample covariance matrix $\widehat{\Sigmam}_\bfy(M)$ associated with the example ASF in \eqref{eq:ex_asf} for different values of $M$ and $N/M=2$.}\label{fig:eig_vec}
\end{figure*}

Once the number of spikes is estimated using MDL, MUSIC proceeds to identify the locations of those spikes. Let $\widehat{\uv}_{\widehat{r}+1,M}, \dots, \widehat{\uv}_{M,M}$ be the eigenvectors in $\widehat{\bfU}$ corresponding to the smallest $M-\widehat{r}$ eigenvalues and let us define $\bfU_\text{noi}=[\widehat{\uv}_{\widehat{r}+1,M}, \dots, \widehat{\uv}_{M,M}]$ as the $M \times (M-\widehat{r})$ matrix corresponding to the noise subspace. The MUSIC objective function is defined as the pseudo-spectrum:
\begin{equation}\label{eq:pseudo_spectrum_1}
    \widehat{\eta}_M (\xi) = \left\|\bfU_\text{noi}^\herm \bfa(\xi)\right\|^2=\sum_{k=\widehat{r}+1}^M \left| \av (\xi)^\herm \widehat{\uv}_{k,M} \right|^2.
\end{equation}
MUSIC estimates the support $\{ \widehat{\phi}_1, \ldots,  \widehat{\phi}_{\widehat{r}}\}$ of the spikes by identifying $\widehat{r}$ dominant minimizers of $\widehat{\eta}_M (\xi)$. It can be shown that for a finite number of spikes $r$, as the number of antennas $M$ and the number of samples $N$ grow to infinity with fixed ratio, MUSIC yields an asymptotically consistent estimate of the spike AoAs. The details are given in Appendix~\ref{sec:MUSIC_analysis} for the sake of completeness. Fig.~\ref{fig:ASF_MUSIC} illustrates the normalized values of the pseudo-spectrum \eqref{eq:pseudo_spectrum_1} for the example ASF in \eqref{eq:ex_asf} and its corresponding sample covariance $\widehat{\Sigmam}_\yv$ for $M=25$ and $N/M=2$. As we can see, the $r=2$ smallest minima of the pseudo-spectrum occur very close to the points $\phi_1=-0.2$ and $\phi_2=0.4$, which are the locations of the spikes in the true ASF. 

\begin{figure}[t]
		\centering
		\includegraphics[width=0.5\linewidth]{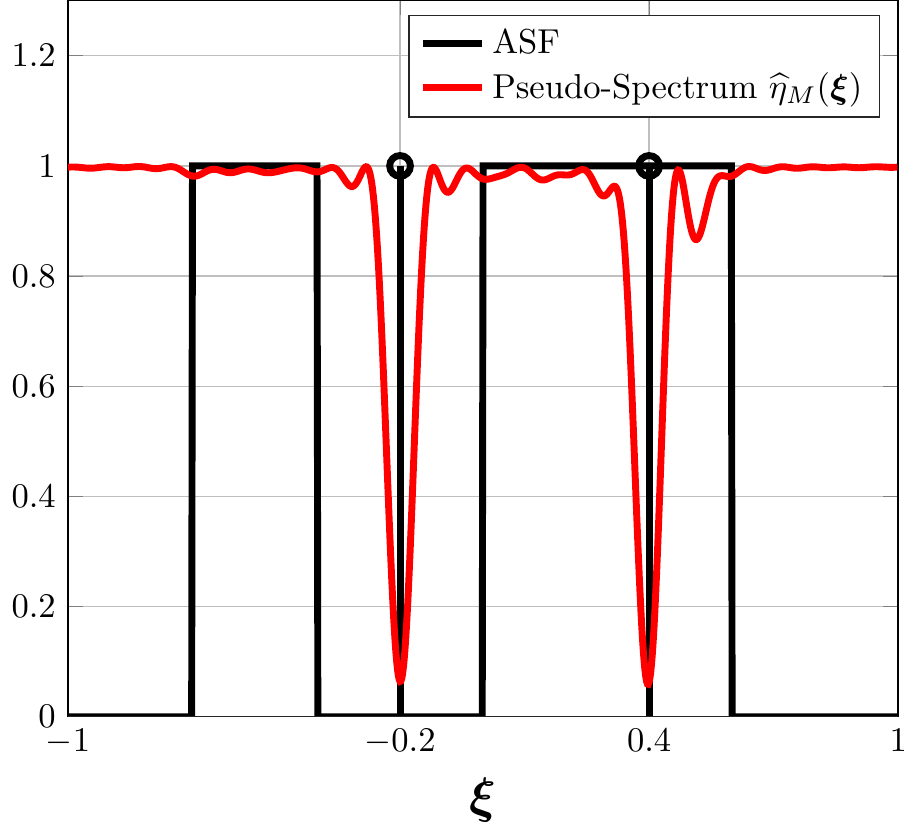}
		\caption{ The pseudo-spectrum plotted for the example ASF in \eqref{eq:ex_asf} with $M=25$ and $N/M=2$.}
		\label{fig:ASF_MUSIC}
\end{figure}

\subsection{Coefficients Estimation by Constrained Least-Squares}  \label{sec:LS}

After obtaining an estimate of the number of spikes and their locations as described before, we need to 
find the spike coefficients $\cv = [c_1,\ldots,c_{\widehat{r}}]^\transp \in \RR_+^{\widehat{r}}$ and the coefficients  $\bv = [b_1,\ldots, b_G]^\transp \in \bR^G$ 
of the  continuous ASF component $\gamma_c (\xi)$ in the form (\ref{dicdic}). 
The dictionary functions $\{\psi_i(\xi)\}$ are selected according to the available prior knowledge about the propagation environment. 
Some typical choices include localized functions $\psi_i(\xi)$, such as Gaussian, Laplacian, or rectangular functions, with a suitably chosen support. 
Fig.~\ref{fig:dic_example} shows an example of Dirac delta and Gaussian dictionaries. 
\begin{figure}[t]
	\centering
	\begin{subfigure}[b]{0.45\textwidth}
		\includegraphics[width=\textwidth]{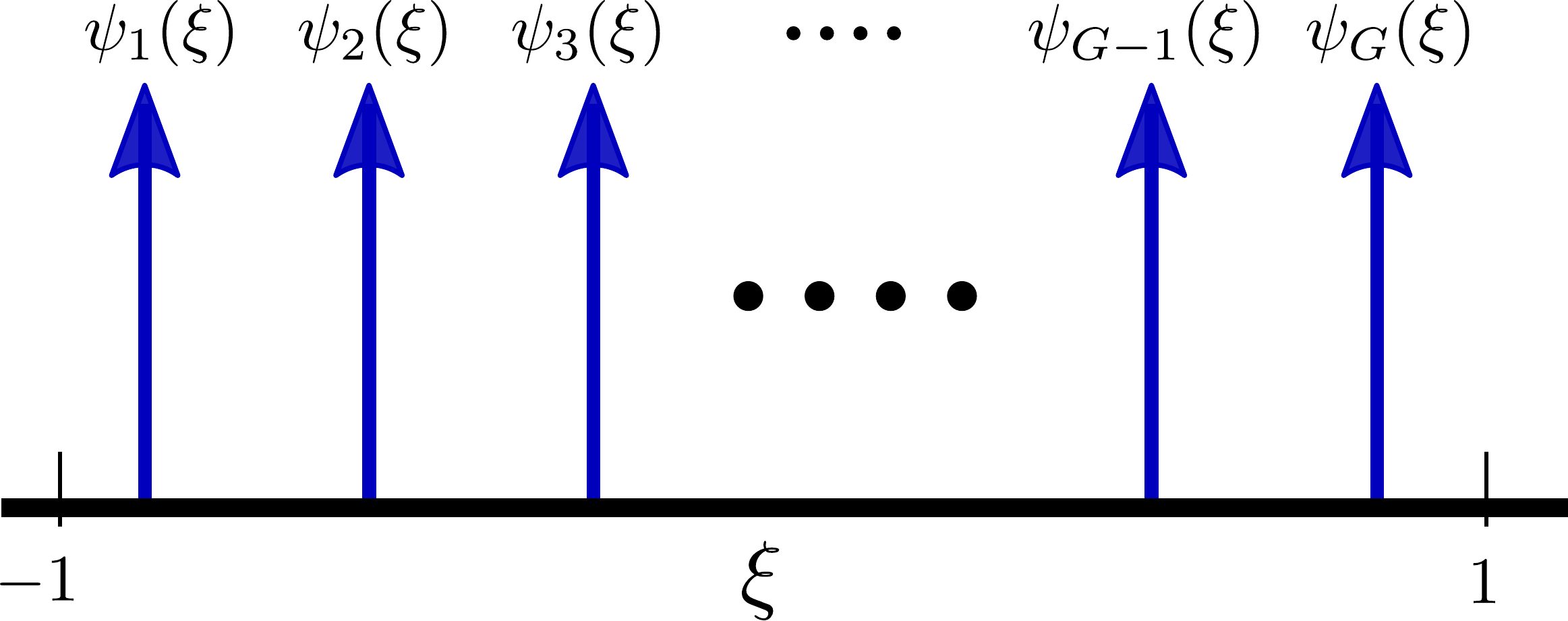}
		\caption{Dirac delta dictionaries}
		\label{fig:kernel_spike}
	\end{subfigure}
	~ 
	\begin{subfigure}[b]{0.45\textwidth}
		\includegraphics[width=\textwidth]{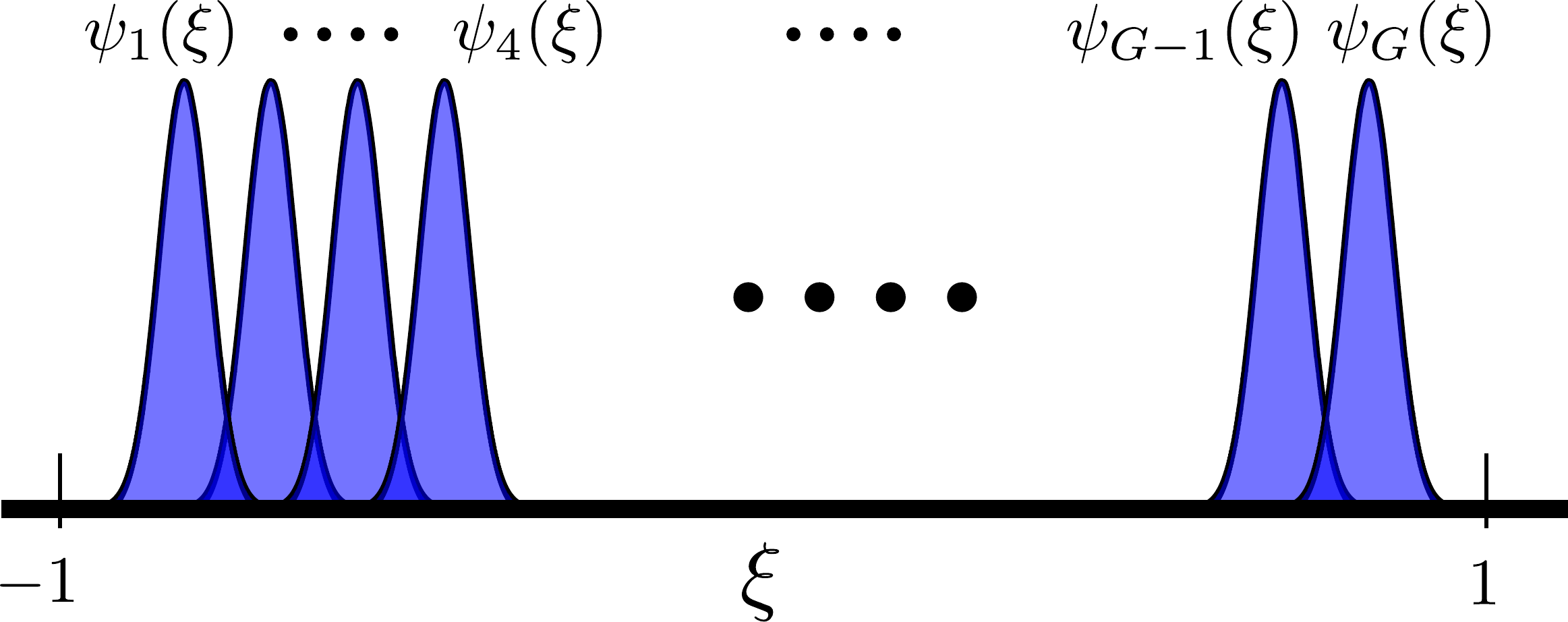}
		\caption{Gaussian dictionaries}
		\label{fig:kernel_gaussian}
	\end{subfigure}
	\caption{ Examples of Dirac delta and Gaussian dictionaries}	
	\label{fig:dic_example}
\end{figure} 
The above representation of the ASF results in the  parametric form of the channel covariance given by 
\begin{equation}\label{eq:decomp}
\begin{aligned}
\SigmaS 
&= \sum_{i=1}^{G+\wh{r}} u_i \Sm_i ,
\end{aligned}
\end{equation}
where we define the model parameter vector $\uv = [u_1,\ldots,u_{G+\widehat{r}}]^\transp = [\bv^\transp,\cv^\transp]^\transp$ and the positive 
semi-definite Hermitian symmetric matrices $\Sm_i =\int_{-1}^1 \psi_{i} (\xi) \av (\xi) \av (\xi)^\herm d \xi,\;\forall i \in [G]$, and $\Sm_{i+G} = \av (\widehat{\phi}_{i}) \av (\widehat{\phi}_{i})^\herm, \;\forall i \in [\wh{r}]$. 


In order to estimate the coefficients vector $\uv$ from the noisy samples $\{\bfy[s]: s\in [N]\}$, we propose three algorithms, namely NNLS estimator, QP estimator and ML-EM estimator. 

\vspace{.2cm}
\hspace{-.5cm}{\em 1) NNLS Estimator}
\vspace{.1cm}

If the dictionary functions have disjoint support, since the overall $\gamma_c(\xi)$ must be non-negative, it follows that the coefficients $\{u_i : i \in [G]\}$ must take values in $\RR_+$, i.e., the whole vector $\uv$ is non-negative.  We know that if the number of samples $N$ is large enough, the sample covariance matrix $\widehat{\Sigmam}_\bfy$ would converge to $\Sigmam_\bfy=\Sigmam_\hv + N_0 \bfI_M$. Then, our goal is to find a good fitting to the sample covariance matrix $\widehat{\Sigmam}_\bfy$ from the set of all covariance matrices of the form 
\begin{align}\label{par_cov_last}
    \Sigmam_\bfy=\Sigmam_\hv(\uv) + N_0 \bfI_M = \sum_{i=1}^{G+\wh{r}} u_i \bfS_i + N_0 \bfI_M.
\end{align}
For this purpose, we use the Frobenius norm as a fitting metric and obtain an estimate of the model coefficients as
\begin{align}
    \bfu^{\star}=\argmin_{\bfu \in \RR_+^{G + \wh{r}}} \; \left\| \widehat{\Sigmam}_\bfy - \sum_{i=1}^{G+\wh{r}} u_i \bfS_i - N_0 \bfI_M\right\|_\sfF^2 = \argmin_{\bfu \in \RR_+^{G + \wh{r}}} \; \left\| \widehat{\Sigmam}_\bfh - \sum_{i=1}^{G+\wh{r}} u_i \bfS_i \right\|_\sfF^2.
\end{align}
Applying vectorization and defining $\Am=[\vec(\Sm_1), \dots, \vec(\Sm_{G+\wh{r}})]$ and $\fv=\vec(\widehat{\Sigmam}_\bfh)$, we can write this as a NNLS  problem
\begin{align}\label{nnls_init_ccp}
    \bfu^{\star}=\argmin_{\bfu \in \RR_+^{G + \wh{r}}}\; \|\bfA \bfu - \bff\|^2,
\end{align}
which can be efficiently solved using a variety of convex optimization techniques (see, e.g., \cite{chen2010nonnegativit,  lawson1974solving}). In our simulation, we use the built-in MATLAB function \textit{lsqnonneg}.

Note that a ULA covariance is Hermitian Toeplitz. By leveraging this, we can reformulate the optimization problem so that the problem dimension is reduced.  Concretely, let $\widetilde{\Sigmam}_{\hv}$ denote the orthogonal projection of  $\widehat{\Sigmam}_\bfh$ onto the space of Hermitian Toeplitz matrices.  This is obtained by averaging the diagonals of $\widehat{\Sigmam}_{\hv}$ and replacing the diagonal elements by the corresponding average  value (see Appendix~\ref{sec:teop_appendix} for completeness). We define the first column of $\widetilde{\Sigmam}_{\hv}$ as $\widetilde{\boldsymbol{\sigma}}$. Then, \eqref{nnls_init_ccp} can be reformulated as 
\begin{align}\label{nnls_toep}
    \bfu^{\star}=\argmin_{\bfu \in \RR_+^{G + \wh{r}}}\; \left\|\Wm\left(\widetilde{\bfA} \bfu - \widetilde{\boldsymbol{\sigma}}\right)\right\|^2,
\end{align}
where $\widetilde{\bfA} = [(\Sm_1)_{\cdot,1},\dots,(\Sm_{G+\wh{r}})_{\cdot,1}]$ is the matrix collecting the first columns $(\Sm_i)_{\cdot,1}$ of the $\Sm_i$ and 
$\Wm = \diag\left(\left[\sqrt{M}, \sqrt{2(M-1)}, \sqrt{2(M-2)}, \dots, \sqrt{2}\right]^{\transp}\right)$ is the weighting matrix to compensate for the number of times an element is repeated in a Hermitian Toeplitz matrix.
\begin{lemma}\label{lemma_1}
    The optimization problems in \eqref{nnls_init_ccp} and \eqref{nnls_toep} are equivalent.
\end{lemma}
\begin{proof}
Note that the difference between \eqref{nnls_init_ccp} and \eqref{nnls_toep} is only the replacement of $\widetilde{\Sigmam}_{\hv}$ from $\widehat{\Sigmam}_{\hv}$. Thus, it is sufficient to show the equivalence of the following two optimization problems: \begin{align}
    \text{P}1: \; \min_{\Xm\in\mathcal{HT}} \|\Xm - \widehat{\Sigmam}_{\hv}\|^2_\sfF, \quad \text{P}2: \; \min_{\Xm\in\mathcal{HT}} \|\Xm - \widetilde{\Sigmam}_{\hv}\|^2_\sfF,
\end{align}
where $\mathcal{HT}$ is the set of Hermitian Toeplitz matrices. Let $\Deltam = \widehat{\Sigmam}_{\hv} - \widetilde{\Sigmam}_{\hv}$. Then, P1 is rewritten as $\underset{\Xm\in\mathcal{HT}}{\min}\;\|\Xm - \widetilde{\Sigmam}_{\hv} - \Deltam\|^2_\sfF$. Since $\widetilde{\Sigmam}_{\hv}$ is the orthogonal projection of $\widehat{\Sigmam}_{\hv}$ on the set $\mathcal{HT}$, by the orthogonality principle the difference $\Deltam = \widehat{\Sigmam}_{\hv} - \widetilde{\Sigmam}_{\hv}$ is orthogonal to the whole set, i.e., $\langle\Deltam, \Tm\rangle := \trace(\Deltam^\herm \Tm) = 0, \; \forall \Tm \in \mathcal{HT}$. It follows that for any $\Xm \in \mathcal{HT}$ the difference $\Xm - \widetilde{\Sigmam}_{\hv}$ is also in $\mathcal{HT}$. Thus
$\|\Xm - \widetilde{\Sigmam}_{\hv} - \Deltam\|^2_\sfF = \|\Xm - \widetilde{\Sigmam}_{\hv}\|^2_\sfF + \|\Deltam|^2_\sfF$ where $\|\Deltam|^2_\sfF$ is constant with respect to $\Xm$ and therefore plays no role in minimization. This proves the equivalence of P1 and P2. 
\end{proof}

\vspace{.2cm}
\hspace{-.5cm}{\em 2) QP Estimator}
\vspace{.1cm}

When the dictionary functions $\psi_i(\xi)$ have overlapping support (e.g., with Gaussian or Laplacian densities), the model coefficients $\{b_i : i = [G]\}$ may take negative values as long as the resulting continuous ASF component is non-negative, i.e.,  $\gamma_c(\xi) = \sum^{G}_{i=1} b_i \psi_i(\xi) \geq 0, \; \forall \xi \in [-1,1]$. We approximate this infinite-dimensional constraint by defining a sufficiently fine grid of  equally spaced points $\{\xi_1,\dots,\xi_{\widetilde{G}}\}$ on $\xi\in[-1,1]$, where $\widetilde{G}$ is generally significantly larger than $G$ and impose the non-negativity of $\gamma_c(\cdot)$ at these points. The resulting constraint is
\begin{align}
    \sum^{G}_{i=1} b_i \psi_i(\xi_j) \geq 0, \quad \forall j \in [\widetilde{G}], \quad
    \iff  \quad \widetilde{\boldsymbol{\Psi}} \bv  \geq 0,
\end{align}
where the elements of the matrix $\widetilde{\boldsymbol{\Psi}} \in \bR^{\widetilde{G}\times G}$ are obtained as $ [\widetilde{\boldsymbol{\Psi}}]_{i,j} = \psi_j(\xi_i), \;\forall i\in[\widetilde{G}], j\in[G]$. Then, the estimation problem under dictionary functions with overlapping support is given by
\begin{equation}\label{eq:generalGaussian}
\begin{aligned}
    \underset{\cv \in \RR_+^{\wh{r}}, \; \bv \in \RR^G}{\text{minimize}} \quad \left\|\Wm\left(\widetilde{\Am} \uv  - \widetilde{\boldsymbol{\sigma}}\right)\right\|^2, \quad
    \text{s.t.} \;\; \widetilde{\boldsymbol{\Psi}} \bv \geq 0,
\end{aligned}
\end{equation}
where $\uv = [\bv^\transp, \cv^\transp]^\transp$ is consistently with the definition in \eqref{eq:decomp}. Notice that \eqref{eq:generalGaussian} is a QP problem and can be solved using standard QP solvers, such as \textit{quadprog} in MATLAB.

\subsection{Coefficients Estimation by Maximum-Likelihood}  \label{sec:ML}

Instead of directly fitting the Frobenius norm between the sample covariance and the parametric covariance, the model parameters can be estimate by the ML method. Give the matrix of the observed noisy channel samples $\Ym$, the likelihood function of $\Ym$ assuming $\Sigmam_\hv = \Sigmam_\hv(\uv)$ in the form of (\ref{eq:decomp}) is given by 
\begin{equation}\label{eq:ML_cals_1}
\begin{aligned}
p \left(\Ym|\uv \right) &= \prod_{s=1}^{N} p \left(\yv[s]|\uv \right) \\
                        & = \prod_{s=1}^{N} \frac{\exp \left( - \yv[s]^\herm \left( \SigmaS + N_0 \mathbf{I}_M \right)^{-1} \yv[s] \right)}{\pi^M \det \left( \SigmaS + N_0 \mathbf{I}_M \right)} \\
                        &= \frac{\exp \left( -\trace\left(\left(\SigmaS + N_0 \mathbf{I}_M \right)^{-1} \Ym \Ym^\herm\right) \right)}{\pi^{MN} \left(\det(\SigmaS + N_0 \mathbf{I}_M)\right)^N }.
\end{aligned} 
\end{equation}
Using \eqref{eq:ML_cals_1} we can form the minus log-likelihood function  $f_{\text{ML}}(\uv) := -\frac{1}{N} \log p \left(\Ym|\uv \right)$. Then, the ML based covariance estimator is obtained by minimizing $f_{\text{ML}}(\uv)$ with respect to the real and non-negative coefficients vector $\uv$, which is formulated as the optimization problem:
\begin{equation}\label{eq:opt_prog}
\begin{aligned}
\underset{\uv \in \bR^{G+\wh{r}}_+}{\text{minimize}}\quad f_{\text{ML}}(\bfu) 
    =&\;\underbrace{\log \det \left(\overset{G+\wh{r}}{ \underset{i=1}{\sum} }  u_i \Sm_i  + N_0 \mathbf{I}_M \right)}_{=f_{\text{cav}}(\uv)} + \underbrace{\trace\left( \left(\overset{G+\wh{r}}{ \underset{i=1}{\sum} } u_i \Sm_i  + N_0 \mathbf{I}_M \right)^{-1}\widehat{\Sigmam}_{\yv}  \right)}_{=f_{\text{vex}}(\uv)},
\end{aligned}
\end{equation}
where $\wh{\Sigmam}_\bfy$ is the sample covariance matrix of the observations defined in \eqref{eq:sample_cov_y}. Note that the objective function $f_{\text{ML}}(\uv) = f_{\text{cav}}(\uv)+f_{\text{vex}}(\uv)$ in \eqref{eq:opt_prog} is the sum of a concave and a convex function and thus \eqref{eq:opt_prog} is not a convex problem. 

It is generally difficult to find the global optimum of a non-convex function such as $f_{\text{ML}}(\bfu)$. A standard approach in such cases is to adopt a MM algorithm \cite{hunter2004tutorial,sun2016majorization}, alternating through two steps with an updating surrogate function that has favorable optimization properties (e.g., convexity) and  approximates the lower-bound of the original objective function. Typical examples of MM algorithms are EM method \cite{dempster1977maximum}, cyclic minimization \cite{journee2010generalized}, and the concave-convex procedure \cite{yuille2002concave}. We choose the EM algorithm to iteratively find a good stationary point of $f_{\text{ML}}(\bfu)$ as we will see that this algorithm yields a computationally efficient update rule and excellent empirical results for the task of estimating the parametric ASF coefficients. Note that although the likelihood function in \eqref{eq:opt_prog} is in a general form for any family of dictionary functions, the EM method can be applied only in the case where all the matrices $\Sm_i$ have rank 1, which is the case when the dictionary functions $\psi_i(\xi)$ are Dirac delta functions. In contrast, the more general concave-convex procedure (e.g., see \cite{khalilsarai2020structured} for the application in this case) can deal with any type of dictionary, but yields significantly higher computational complexity, so that it is not suited for large $M$ (massive MIMO case). In this work, we deal with general dictionary functions using constrained LS, while restrict the use of EM to the case of Dirac delta dictionary functions. 

{\bf Application of the EM algorithm.}
When $\psi_i(\xi) = \delta(\xi - \xi_i)$ where $\{\xi_i : i \in [G]\}$ are uniformly spaced points in $[-1,1)$, 
we have $\Sm_i = \av(\xi_i) \av^\herm(\xi_i)$ for $i \in [G]$ and $\Sm_{G+i} = \av(\widehat{\phi}_i) \av^\herm(\widehat{\phi}_i)$ for $i \in [\wh{r}]$. 
Then, defining the extended grid $\{ \xi_i : i \in [G + \wh{r}]\}$ with $\xi_{G+i} = \widehat{\phi}_i$ for $i \in [\wh{r}]$, the parametric form of channel covariance $\Sigmah$ can be written as
\begin{equation}\label{eq:deltaCovariance}
    \Sigmah(\uv) = \sum_{i=1}^{G+\wh{r}} u_i \Sm_i = \Dm \Um \Dm^\herm,
\end{equation}
where $\Dm = [\av(\xi_1), \dots, \av(\xi_{G+\wh{r}})]$ and $\Um = \diag(\uv)$. Then,  $\Sigmam_{\yv} = \Sigmah + N_0\bfI_M $  can be formally considered as the covariance  matrix of the approximated received samples
\begin{equation}\label{eq:channel_EM}
    \yv[s] = \sum^{G+\wh{r}}_{i=1}\rho_i[s] \av(\xi_i) + \zv[s] = \Dm \xv[s] + \zv[s],
\end{equation}
where $\xv[s]= [\rho_1[s], \dots, \rho_{G+\wh{r}}[s]]^\transp$ is the latent variable vector containing the instantaneous random path gain, and we have $\xv[s]\sim \mathcal{CN}(\mathbf{0},\Um) ,\forall s \in [N]$ by considering $\uv\in\mathbb{R}_+^{G+\wh{r}}$ as the Gaussian prior variance element-wise corresponding to the Gaussian process $\xv[s]$. 

Under the formulation in \eqref{eq:channel_EM}, the $N$ noisy samples collected as columns of the matrix $\Ym$ can be written as $\Ym = \Dm \Xm + \Zm$, where $\Xm = [\xv[1],\dots,\xv[N]]$ and $\Zm = [\zv[1],\dots,\zv[N]]$. The EM algorithm treats $(\Ym, \Xm)$ as the {\em incomplete data}, where $\Xm$ is referred to as {\em missing data}. Using the fact that $\Ym$ given $\Xm$ is Gaussian with mean $\Dm \Xm$ and independent components with variance $N_0$, we have that $p(\Ym,\Xm|\uv) = p(\Ym|\Xm) p(\Xm|\uv)$, where $p(\Ym|\Xm)$ is a conditional Gaussian distribution that does not depend on $\uv$. By marginalizing with respect to $\Xm$ and taking the logarithm, the log-likelihood function takes on the form of a conditional expectation $\Lc(\uv)  := \log p(\Ym | \uv) =  \log \EE_{\Xm|\uv} [ p(\Ym|\Xm) ]$. The EM algorithm maximizes iteratively the lower bound of $\Lc(\uv)$ by alternating the expectation step (E-step) and the maximization step (M-step) \cite{sun2016majorization} as follows. Let $\widehat{\uv}^{(\ell)}$ be the estimate of $\uv$ in the $\ell$-th iteration, by introducing a posterior density of $\Xm$ as $p(\Xm|\Ym,\widehat{\uv}^{(\ell)})$ we have
\begin{align}
    \Lc(\uv)&=  \log \EE_{\Xm|\uv} \left[  \frac{p(\Ym|\Xm)p(\Xm|\Ym,\widehat{\uv}^{(\ell)})}{p(\Xm|\Ym,\widehat{\uv}^{(\ell)})}  \right] 
    = \log \EE_{\Xm|\Ym,\widehat{\uv}^{(\ell)}} \left[  \frac{p(\Ym|\Xm)p(\Xm|\uv)}{p(\Xm|\Ym,\widehat{\uv}^{(\ell)})}  \right] \\
    & \overset{(a)}{\geq}  \EE_{\Xm|\Ym,\widehat{\uv}^{(\ell)}} [\log p(\Ym| \Xm) + \log p(\Xm| \uv) - \log p(\Xm|\Ym,\widehat{\uv}^{(\ell)}) ] := \widetilde{\Lc}(\uv|\widehat{\uv}^{(\ell)}), \label{eq:LUU}
\end{align}
where $(a)$ follows Jensen's inequality and the concavity of $\log(\cdot)$ and gives the lower bound $\widetilde{\Lc}(\uv|\widehat{\uv}^{(\ell)})$ to the log-likelihood function $\Lc(\uv)$.  The E-step consists of computing $\widetilde{\Lc}(\uv|\widehat{\uv}^{(\ell)})$.
Using the joint conditional Gaussianity of $\Ym$ and $\Xm$ given $\uv = \widehat{\uv}^{(\ell)}$, $\widetilde{\Lc}(\uv|\widehat{\uv}^{(\ell)})$ can be evaluated in closed form by computing the conditional mean and covariance of $\xv[s]$ given $\yv[s]$ and $\widehat{\uv}^{(\ell)}$, respectively given by \cite{wipf2004sparse}
\begin{align}\label{eq:EM_mu_sigma}
    \text{E-step:}\quad \boldsymbol{\mu}^{(\ell)}_{\xv[s]} = \frac{1}{N_0}\Sigmam^{(\ell)}_{\xv} \Dm^\herm\yv[s],\;\;
    \Sigmam^{(\ell)}_{\xv} = \left(\frac{1}{N_0}\Dm^\herm\Dm + \left(\widehat{\Um}^{(\ell)}\right)^{-1}\right)^{-1},
\end{align}
where we define  $\widehat{\Um}^{(\ell)} = \diag(\widehat{\uv}^{(\ell)})$.  The M-step consists of computing
\begin{equation}\label{eq:EM_solution}
    \widehat{\uv}^{(\ell+1)} = \argmax_{\uv \in \RR_+^{G + \wh{r}}} \; \widetilde{\mathcal{L}}(\uv|\widehat{\uv}^{(\ell)}).
\end{equation}
Note that $p(\Ym|\Xm)$ and $p(\Xm|\Ym,\widehat{\uv}^{(\ell)})$ in \eqref{eq:LUU} do not depend on $\uv$ and thus can be neglected in the M-step. We find the argument of the maximization in the simplified form (details are omitted for brevity)
\begin{equation}\label{eq:EM-cost}
    \doublewidetilde{\mathcal{L}}(\uv|\widehat{\uv}^{(\ell)})  := \EE_{\Xm|\Ym,\widehat{\uv}^{(\ell)}}[\log p (\Xm|\uv)]=\sum^{G+\wh{r}}_{i=1}\left(-N\log(\pi u_i) - \frac{\sum_{s=1}^{N}\left|\left[\boldsymbol{\mu}_{\xv[s]}^{(\ell)}\right]_i\right|^2+ N\left[\Sigmam^{(\ell)}_{\xv} \right]_{i,i}}{u_i}  \right).
\end{equation}
It is observed from \eqref{eq:EM-cost} that the maximization is decoupled with respect to each component $u_i$ of $\uv$.  
Then, the optimality in the $\ell$-th iteration is also easily obtained in closed form. Setting each partial derivative $\frac{\partial}{\partial u_i}  \doublewidetilde{\mathcal{L}}(\uv|\widehat{\uv}^{(\ell)})$ for $i = [G + \widehat{r}]$ to zero,  we find
\begin{align}
    \text{M-step:}\quad\widehat{u}^{(\ell+1)}_i
    &= \frac{1}{N}\sum^{N}_{s=1}\left|\left[\boldsymbol{\mu}_{\xv[s]}^{(\ell)}\right]_i\right|^2 + \left[\Sigmam^{(\ell)}_{\xv}\right]_{i,i}, \quad \forall i \in [G+\wh{r}].\label{eq:M-step}
\end{align}
With a initial point $\widehat{\uv}^{(0)}$ the ML-EM algorithm iteratively runs the E-step and M-step  until the stop condition  $f_{\text{ML}}(\widehat{\uv}^{(\ell)})-f_{\text{ML}}(\widehat{\uv}^{(\ell+1)})\leq \epsilon_{\text{EM}}$ is met, where $\epsilon_{\text{EM}}$ is the predefined stop threshold. The initial point is usually set as an all ones vector. In our case, it can be set as the result of NNLS solution. 
An extensive comparison of these two initializations obtained by simulating several different channel scattering geometries is depicted in Fig.~\ref{fig:EM_convergence}, which shows that both initializations converge within 100 iterations. It reveals that  the
NNLS initialization converges much faster and results in lower values of the objective function. Therefore, in our results we used the NNLS solution to initialize ML-EM algorithm. 

\begin{remark} 
The use of Dirac delta functions as dictionary functions to approximate the continuous part of the ASF $\gamma_c$ may seems a contradiction, since 
by definition this component does not contain spikes. However, there is a fundamental difference between the Dirac components corresponding to spikes and
the equally spaced ``picket-fence'' used to approximate $\gamma_c$. This can be noticed by observing that the power associated to 
the $i$-th spike component is $c_i$, which is a constant independent on the number of grid points $G$, while the power associated to 
the $i$-th dictionary function $\psi_i(\xi) = \delta(\xi-\xi_i)$ is $b_i =  2\gamma_c(\xi_i)/G$ where $2/G$ is the spacing of the uniform grid on $[-1,1)$. 
For sufficiently large $G$, the smooth function $\gamma_c$ can be approximated by the picket-fence of scaled Dirac deltas in the sense that, 
for any sufficiently smooth test function $f(\xi)$ ($\Lc_1$-integrable on $[-1,1]$, piecewise continuous, with at most a countably infinite number of discontinuities) 
we have $\lim_{G \rightarrow \infty} \int_{-1}^1 ( \gamma_c(\xi) -  \frac{2}{G} \sum_{i=1}^G \gamma_c(\xi_i) \delta(\xi - \xi_i) ) f(\xi)  d\xi = 0$. 
\hfill $\lozenge$ 
\end{remark}

\begin{figure}[t]
		\centering
		\includegraphics[width=0.5\linewidth]{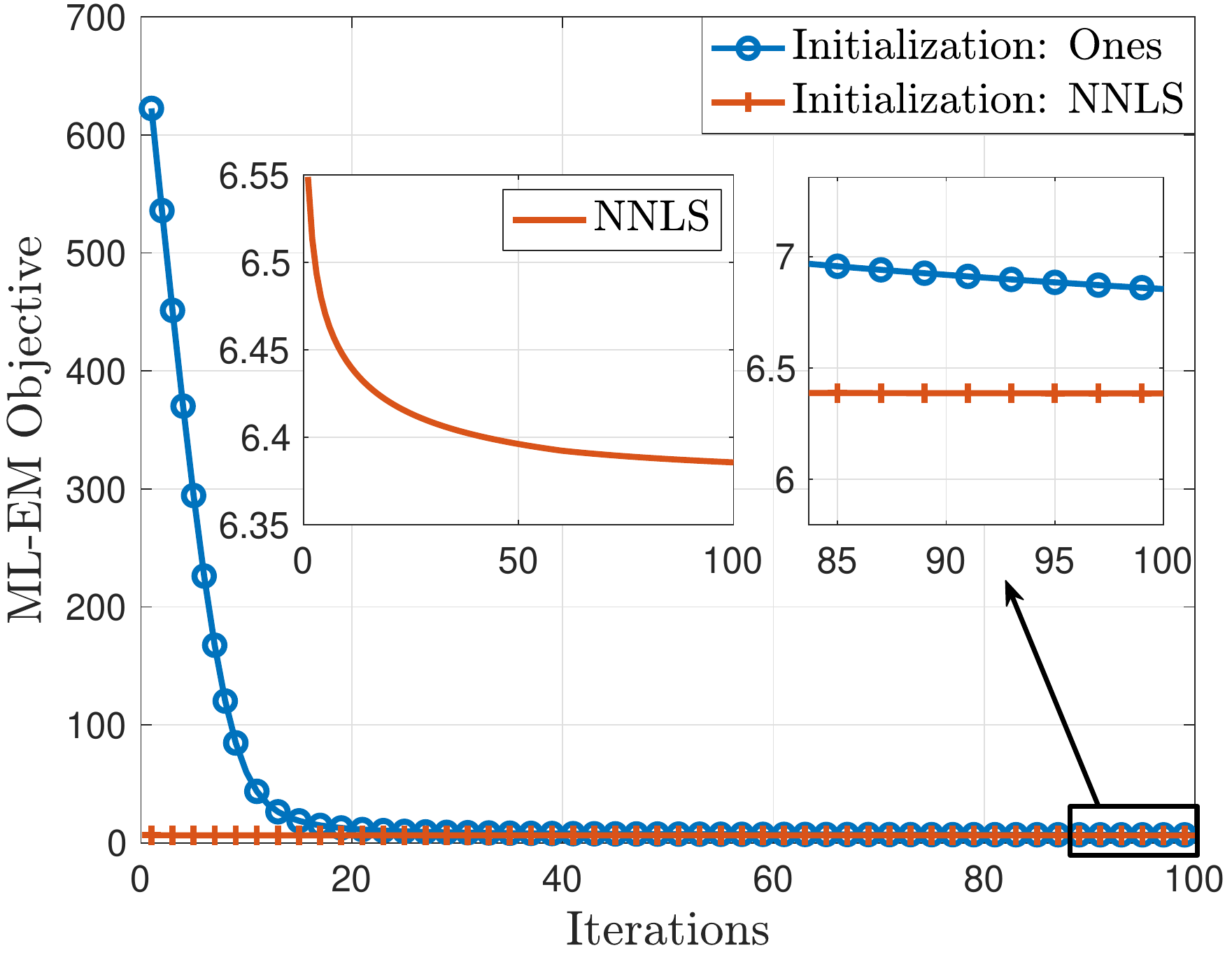}
		\caption{Convergence behaviour of ML-EM algorithm with NNLS and all ones initializations.}
		\label{fig:EM_convergence}
\end{figure}

\section{Numerical Results}\label{sec:simulations}
In this section, the proposed constrained LS and ML-EM based estimators are numerically evaluated and compared with
existing state of the art methods. We use the realistic channel emulator QuaDriGa \cite{jaeckel2014quadriga} to generate the channel samples. We adopt two communication scenarios in QuaDriGa: 3GPP 3D Urban Macro-Cell Line Of Sight (3GPP-3D-UMa-LOS) and 3GPP 3D Urban Macro-Cell Non-Line Of Sight (3GPP-3D-UMa-NLOS). The BS array adopts horizontal ULA with $M=128$ antennas. The carrier frequency of the UL and DL are 1.9 GHz and 2.1 GHz, respectively.  The SNR is set to $10$ dB. The results are averaged over 20 random ASFs and 100 times of random channel realizations for each ASF.

\subsection{Compared Benchmarks}
We compare the proposed algorithms to the following three benchmarks.

\begin{enumerate}
    \item \textit{Toeplitz-PSD Projection}: The first benchmark is an intuitively simple approach. We know that the ULA channel covariance  is a Toeplitz PSD matrix. Thus, we can project the sample covariance onto the space of Toeplitz-PSD by solving the convex optimization problem:
    \begin{align}\label{eq:Toeplitz-PSD}
        \Sigmah^{\text{PSD}} = \argmin_{\Sigmam \in\mathcal{HT} } \; \|\Sigmam - \widehat{\Sigmam}_{\hv}\|^2_\sfF, \quad \text{s.t.} \; \Sigmam\succeq \mathbf{0}.
    \end{align}
    The projected matrix $\Sigmah^{\text{PSD}}$ is the covariance estimate. 
    We use the projections onto convex sets (POCS) algorithm \cite{bauschke1996projection} to solve \eqref{eq:Toeplitz-PSD}. Specifically, the sample covariance matrix is alternately projected onto the convex set of toeplitz and PSD cone. Projecting on toeplitz follows the way in Appendix~\ref{sec:teop_appendix}, and projection on the PSD cone is achieved by eigen-decomposition and setting all negative eigenvalues to zero.  These two projections are repeated till convergence. Note that the complexity of this semi-definite programming problem can be high when the number of antennas $M$ is large. Moreover, toeplitz-PSD can \textit{\textbf{not}} provide UL-DL covariance transformation.

    \item \textit{The SPICE Method}: The second method we use for comparison is known as sparse iterative covariance-based estimation (SPICE) \cite{stoica2011spice}. This method also exploits the ASF domain but can be only applied with Dirac delta dictionaries. Similar to the parametric covariance model with only Dirac delta dictionaries introduced in \eqref{eq:deltaCovariance}, assuming Dirac delta dictionaries $\Dm =[\av (\xi_1),\ldots,\av (\xi_G)]$ of $G$ array response vectors corresponding to $G$ AoAs, and defining $\Sigmam = \Dm \diag(\uv) \Dm^\herm$, the ASF coefficients $\uv$ is estimated by solving the following convex optimization problem for two cases:
    \begin{equation}\label{eq:SPICE}
        \uv^\star = \begin{cases} 
                \underset{\uv\in \bR_+^{G}}{\argmin} \; \left \Vert \Sigmam^{-1/2} \left(\widehat{\Sigmam}_\yv - \Sigmam\right) \right\Vert^2_\sfF, & N < M, \\
                \underset{\uv\in \bR_+^{G}}{\argmin} \; \left \Vert \Sigmam^{-1/2} \left(\widehat{\Sigmam}_\yv - \Sigmam\right)\widehat{\Sigmam}_\yv^{-1/2} \right\Vert^2_\sfF, & N\geq M.
            \end{cases}
    \end{equation}
    The channel covariance estimate is then obtained as $\Sigmam_\hv^{\text{SPICE}} = \Dm \diag(\uv^\star) \Dm^\herm$.

    \item \textit{Convex Projection Method}: This method is proposed in \cite{miretti2018fdd} for the ASF estimation by solving a convex feasibility problem $\widehat{\gamma} = \text{find} \; \gamma, \text{subject to} \; \gamma \in \mathcal{S}$, where
    \begin{equation}
        \mathcal{S} = \left\{ \gamma : \int^1_{-1} \gamma(\xi)e^{j\pi m \xi}d\xi= [\widehat{\boldsymbol{\Sigma}}_{\hv}]_{m,1},\;m=[M],\gamma(\xi)\geq0, \forall \xi \in [-1,1]\right\}.
    \end{equation}
    This can be solved by applying an iterative projection algorithm, which produces a sequence of functions in $L_2$ that converges to a function satisfying the constraint $\gamma \in \mathcal{S}$. Given the estimated ASF $\widehat{\gamma}$, the channel covariance estimation is obtained following   \eqref{eq:cov_mat}.

\end{enumerate}

\subsection{Considered Metrics}
Denoting a generic covariance estimate as $\widehat{\Sigmam}$, we use three error metrics to evaluate the estimation quality:
\begin{enumerate}
    
    \item \textit{Normalized Frobenius-norm Error}: This error is defined as 
	 \begin{equation}
	    E_{\text{NF}} =   \frac{\Vert \Sigmah -  \widehat{\Sigmam} \Vert_\sfF}{\Vert \Sigmah \Vert_\sfF}.
	 \end{equation}
	
	\item \textit{Normalized MSE of Channel Estimation}: Given a noisy channel observation $\yv = \hv + \zv$, the optimal estimation of channel vector $\hv$ is obtained via linear MMSE filter $\widehat{\hv} = \widehat{\Sigmam}(N_0\mathbf{I}_M + \widehat{\Sigmam})^{-1}\yv$. Then, this metric considers the normalized mean squared error (NMSE) of instantaneous channel estimation:
	\begin{equation}
	    E_{\text{NMSE}} = \frac{\EE[\|\hv - \widehat{\hv}\|^2]}{\EE[\|\hv \|^2]}.
	\end{equation}
	By the optimality of linear MMSE estimation for Gaussian random vectors, 
	the lower bound of this error is obtained using the true channel covariance.
	
	\item \textit{Power Efficiency}: This metric evaluates the similarity of dominant subspaces between the estimated and true matrices, which is an important factor in various applications of massive MIMO such as user grouping and group-based beamforming \cite{adhikary2013joint,nam2014joint,khalilsarai2018fdd}. Specifically, let $p\in[M]$ denote a subspace dimension parameter and let  $\Um_p \in \bC^{M\times p}$ and $\widehat{\Um}_p\in \bC^{M\times p}$ be the $p$ dominant eigenvectors of $\Sigmah$ and $\widehat{\Sigmam}$ corresponding to their largest $p$ eigenvalues, respectively. The power efficiency (PE) based on $p$ is defined as
	\begin{equation}
	    E_{\text{PE}}(p) =  1 - \frac{\langle\Sigmah,\widehat{\Um}_p\widehat{\Um}_p^\herm\rangle}{\langle\Sigmah,\Um_p\Um_p^\herm\rangle} =  \frac{\trace\left(\Um_p^\herm\Sigmah\Um_p\right) - \trace\left(\widehat{\Um}_p^\herm\Sigmah\widehat{\Um}_p\right)}{\trace\left(\Um_p^\herm\Sigmah\Um_p\right)}. 
	\end{equation}
	It is noticed that $E_{\text{PE}}(p) \in [0,1]$ and the closer it is to 0, the more power is captured by the estimated  $p$-dominant subspace. 
	$E_{\text{PE}} = 0$ is achieved when $\widehat{\Um}_p = \Um_p$, i.e., perfect estimation. 
	
\end{enumerate}

\subsection{Performance Comparison}
We first provide a performance comparison under only Dirac delta dictionaries with a relatively large number $G$ of grid points. 
Then, we provide a comparison under different dictionaries to show the benefit of properly choosing the dictionary adapted to the case at hand. 

\subsubsection{Comparison under Dirac delta dictionaries}\label{sec:asf-dirac}
We set the number of dictionary functions for the ASF continuous part as $G=2M$. 
SPICE is also applied to the same picket-fence dictionary without knowledge of the spike locations, 
since this is a feature of our own method and not intrinsic in the SPICE algorithm. 
The projection method is not a dictionary based method. 
To implement it, we discretize the ASF domain with 5000 grid points to approximate the continue ASF.   
\newcommand{\plotwidth}{0.31}
\begin{figure*}[t]
	\centering
	\begin{subfigure}[b]{\plotwidth\textwidth}
		\includegraphics[width=\textwidth]{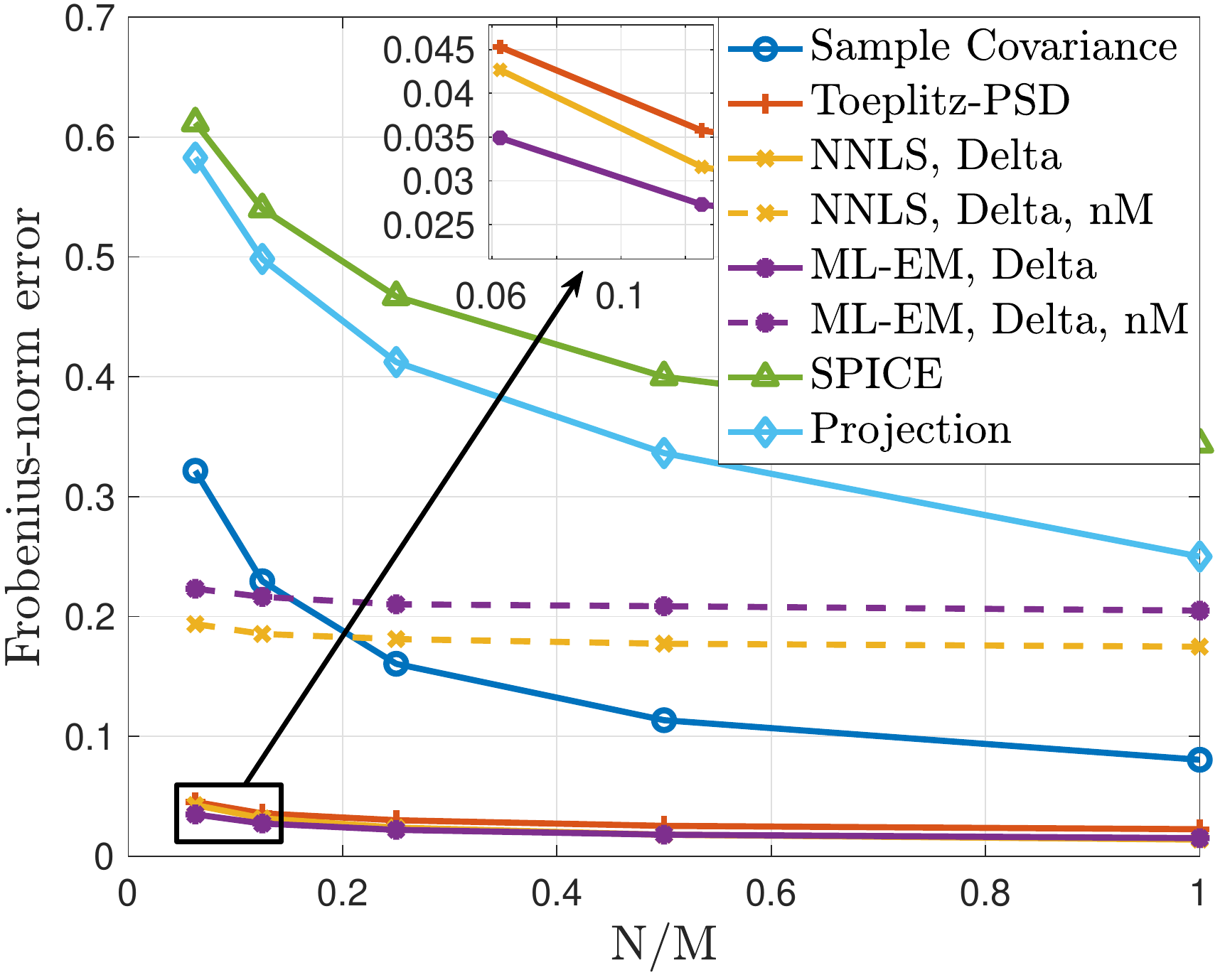}
		\caption{UL Frobenius-norm error}
		\label{fig:fro}
	\end{subfigure}
	~ 
	\begin{subfigure}[b]{\plotwidth\textwidth}
		\includegraphics[width=\textwidth]{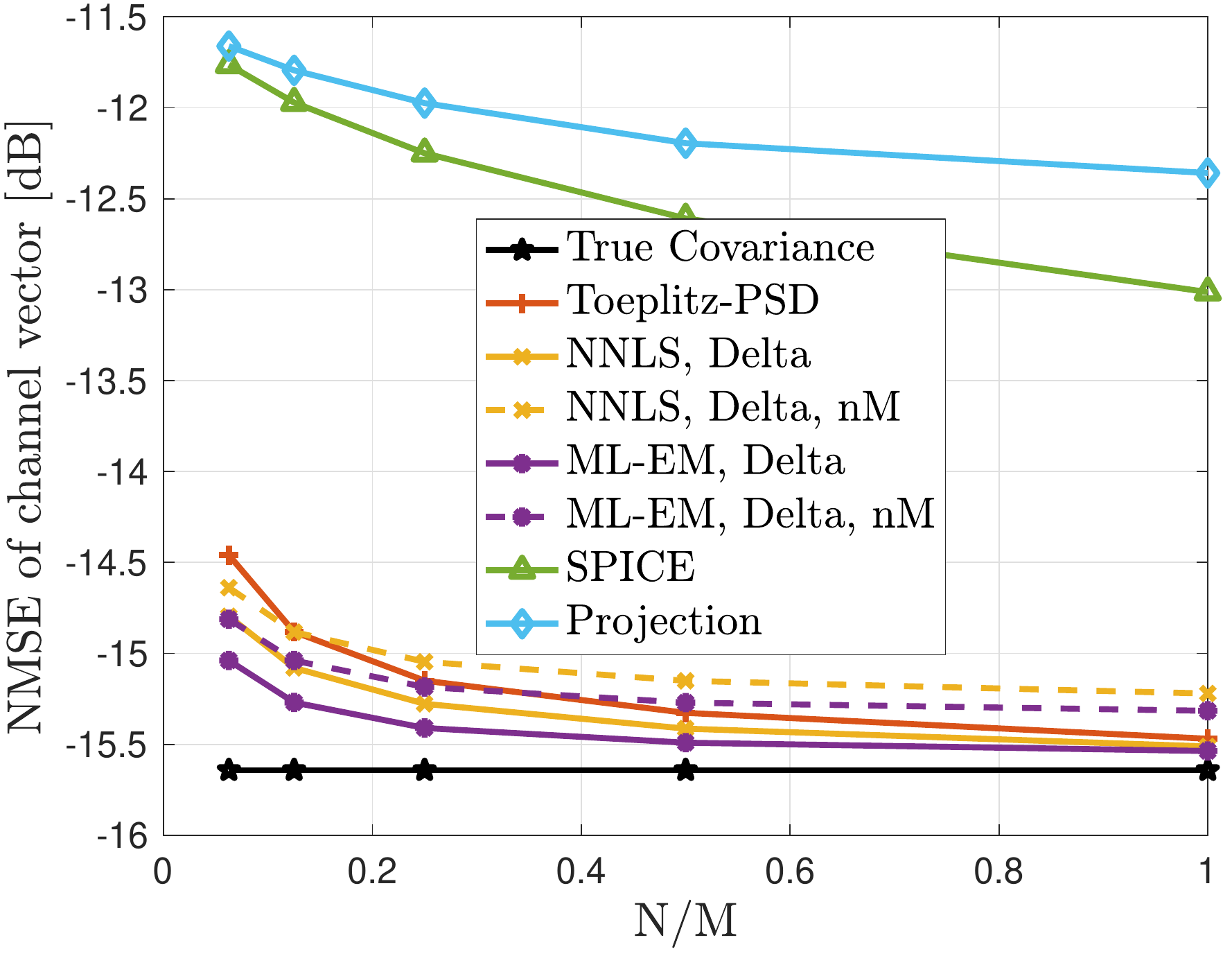}
		\caption{UL MSE of channel estimation}
		\label{fig:MSE}
	\end{subfigure}
	~
	\begin{subfigure}[b]{\plotwidth\textwidth}
		\includegraphics[width=\textwidth]{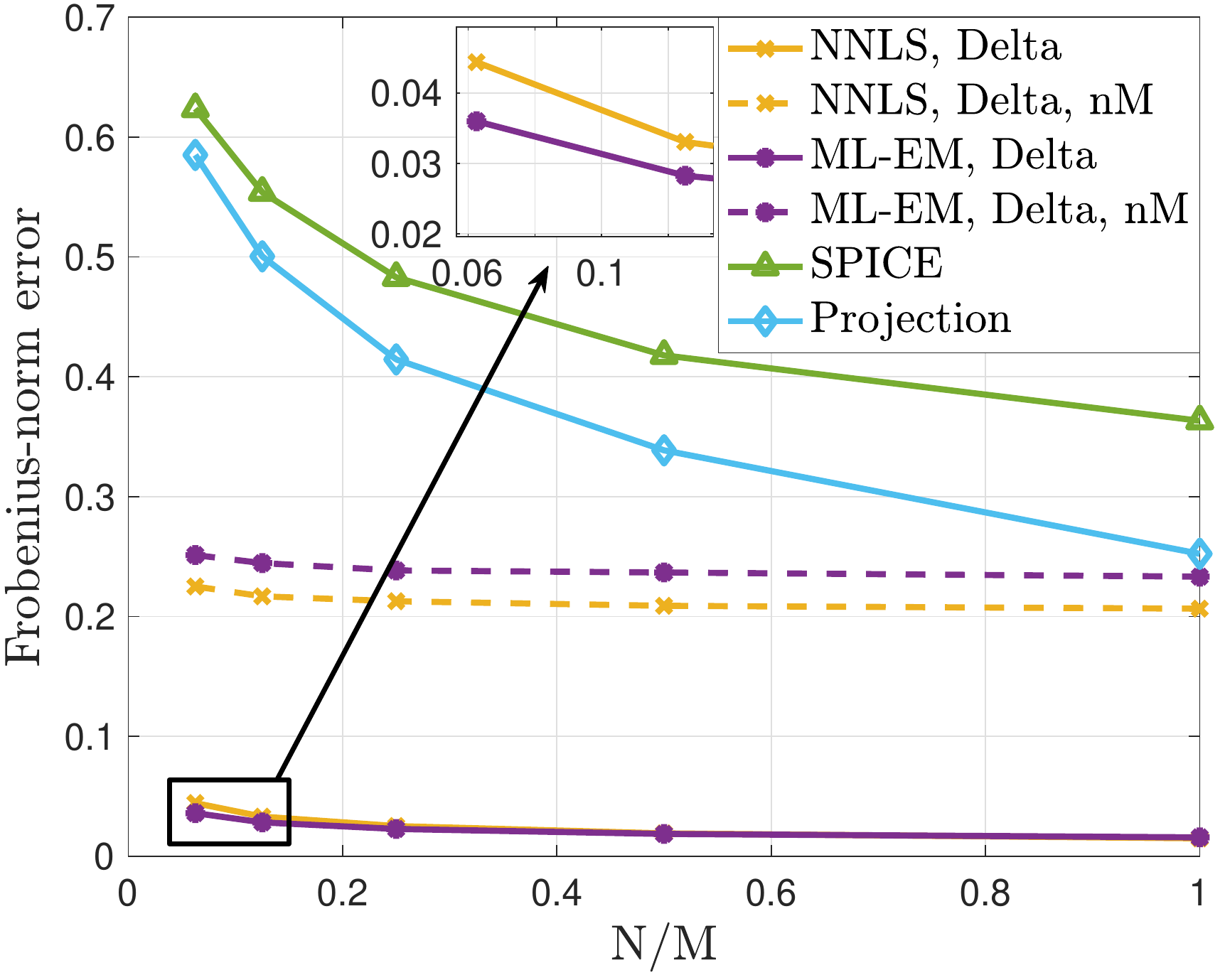}
		\caption{DL Frobenius-norm error}
		\label{fig:fro_dl}
	\end{subfigure}
	~
	\begin{subfigure}[b]{\plotwidth\textwidth}
		\includegraphics[width=\textwidth]{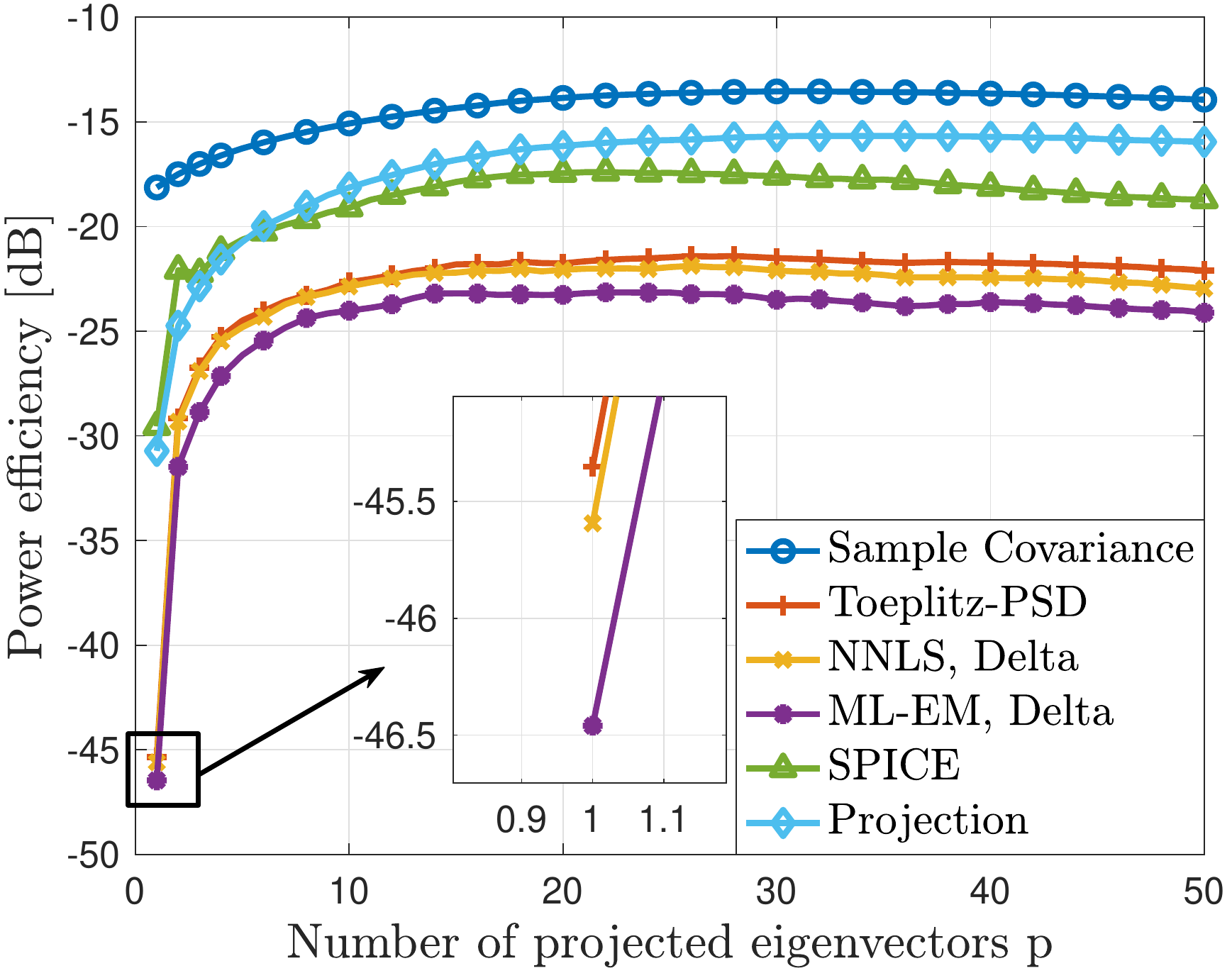}
		\caption{UL PE with $N/M=0.125$}
		\label{fig:PE1}
	\end{subfigure}
	~
	\begin{subfigure}[b]{\plotwidth\textwidth}
		\includegraphics[width=\textwidth]{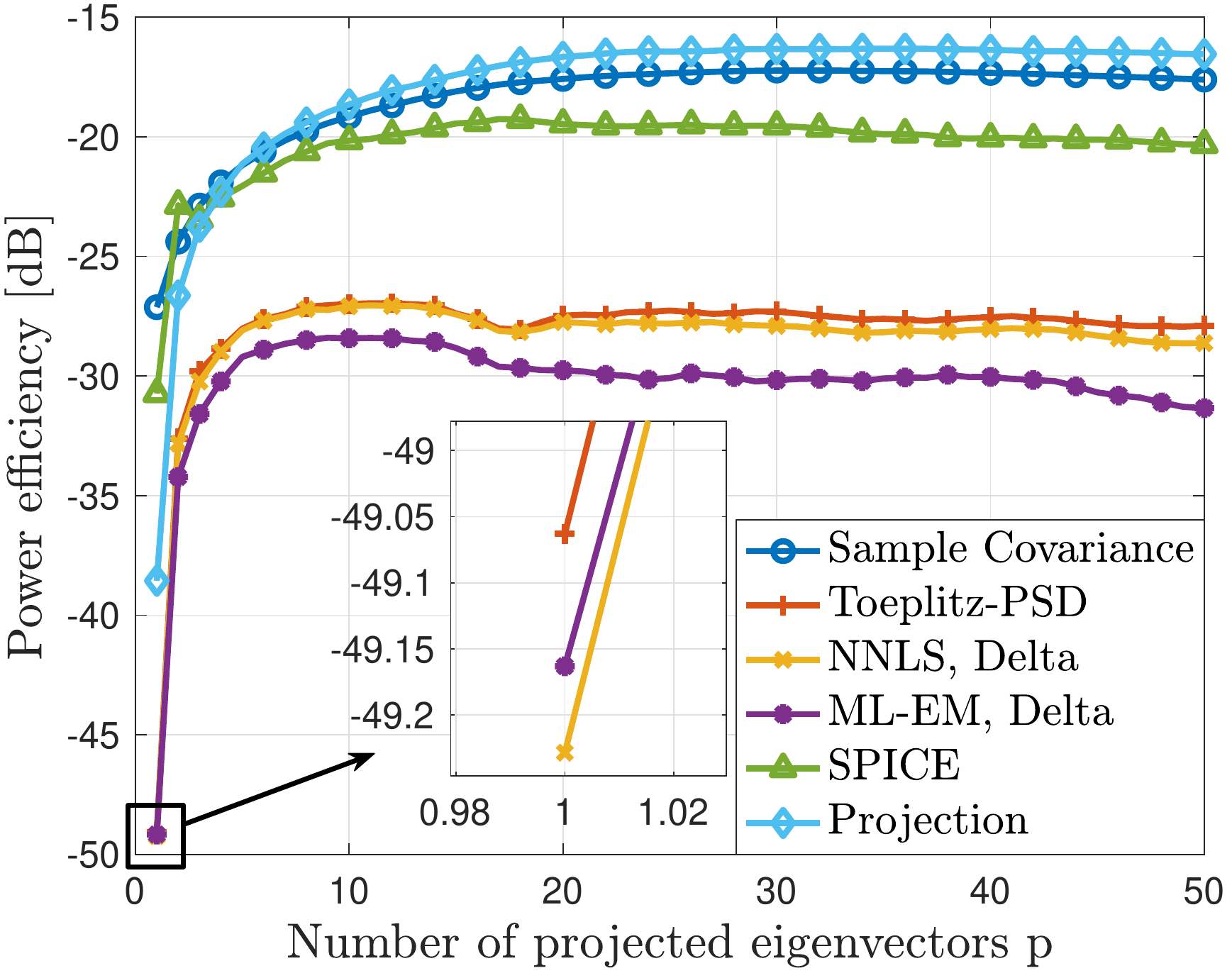}
		\caption{UL PE with $N/M=0.5$}
		\label{fig:PE2}
	\end{subfigure}
	~
	\begin{subfigure}[b]{\plotwidth\textwidth}
		\includegraphics[width=\textwidth]{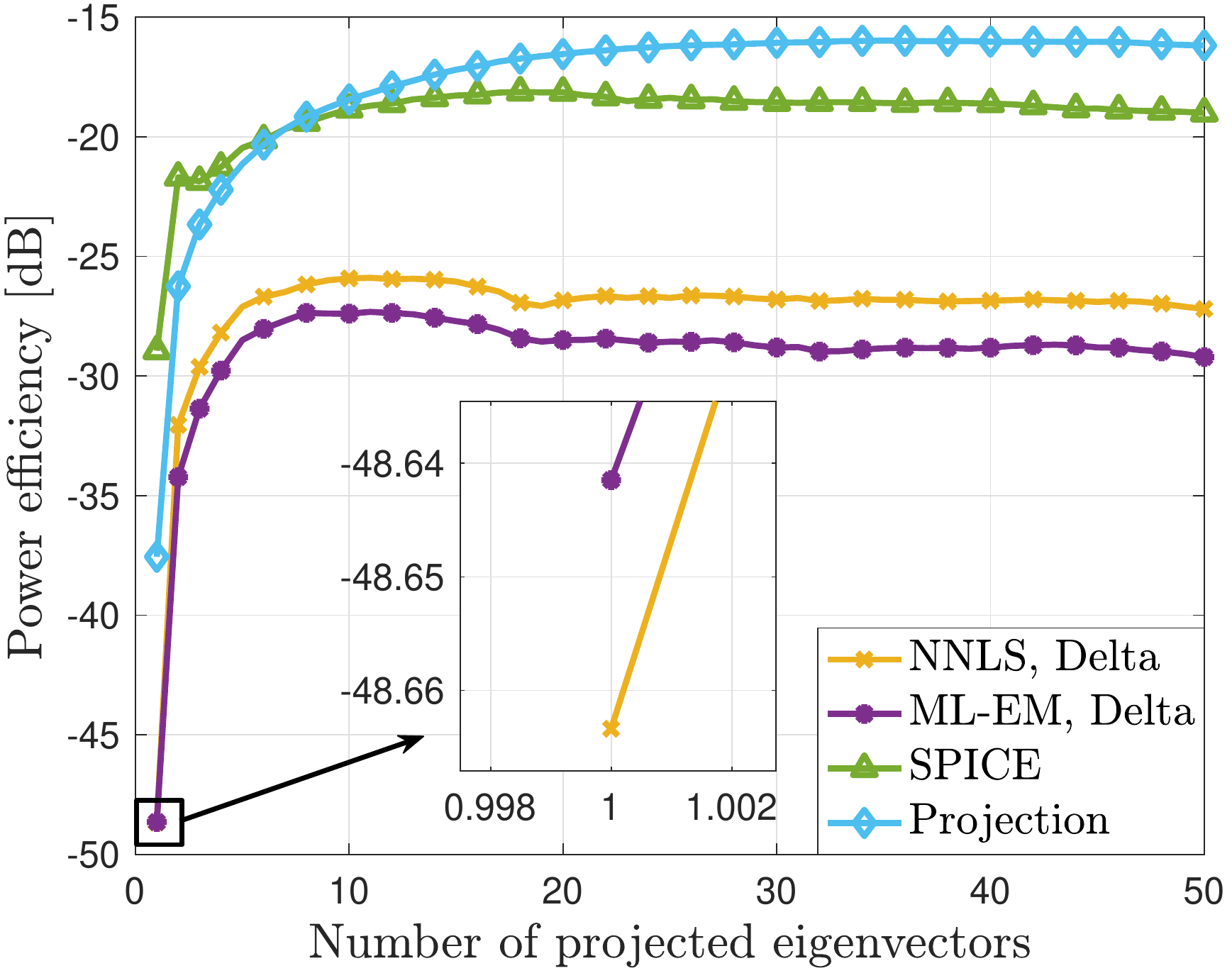}
		\caption{DL PE with $N/M=0.5$}
		\label{fig:PE_dl}
	\end{subfigure}
	\caption{ Covariance estimation quality comparison for scenario 3GPP-3D-UMa-LOS.} \label{fig:spike_err}
\end{figure*} 

\begin{figure*}[t]
	\centering
	\begin{subfigure}[b]{\plotwidth\textwidth}
		\includegraphics[width=\textwidth]{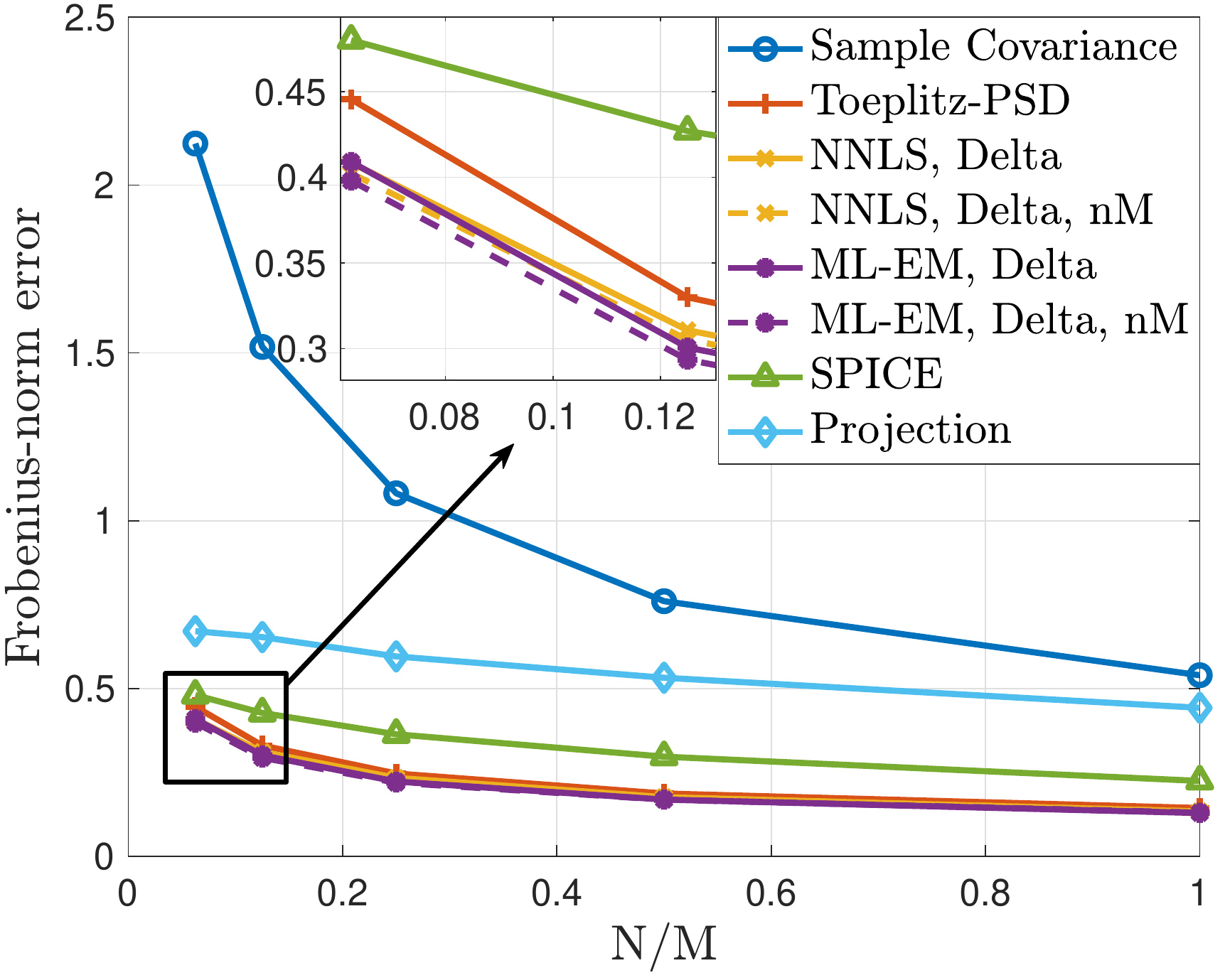}
		\caption{UL Frobenius-norm error}
		\label{fig:NLOS-fro}
	\end{subfigure}
	~ 
	\begin{subfigure}[b]{\plotwidth\textwidth}
		\includegraphics[width=\textwidth]{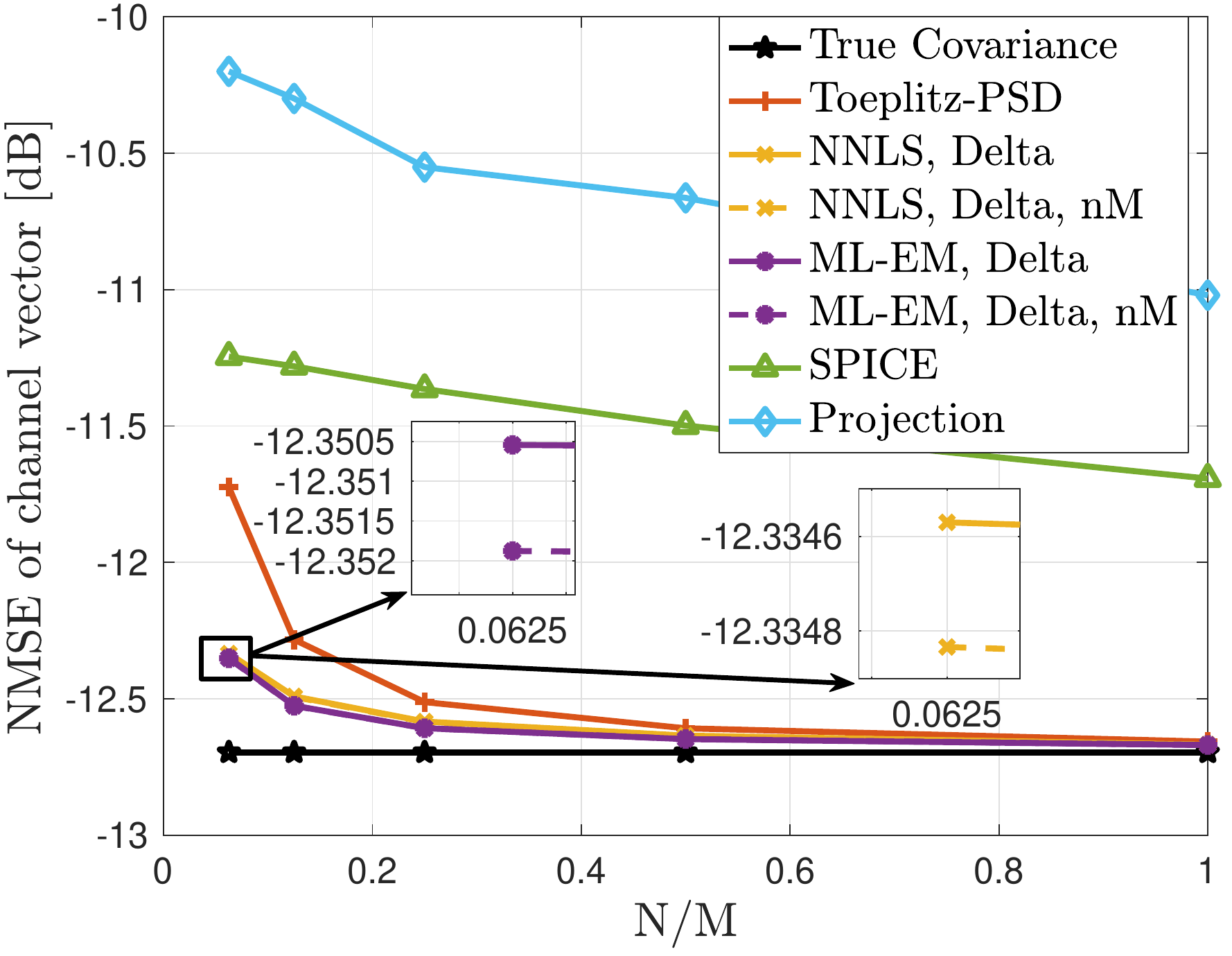}
		\caption{UL MSE of channel estimation}
		\label{fig:NLOS-MSE}
	\end{subfigure}
	~
	\begin{subfigure}[b]{\plotwidth\textwidth}
		\includegraphics[width=\textwidth]{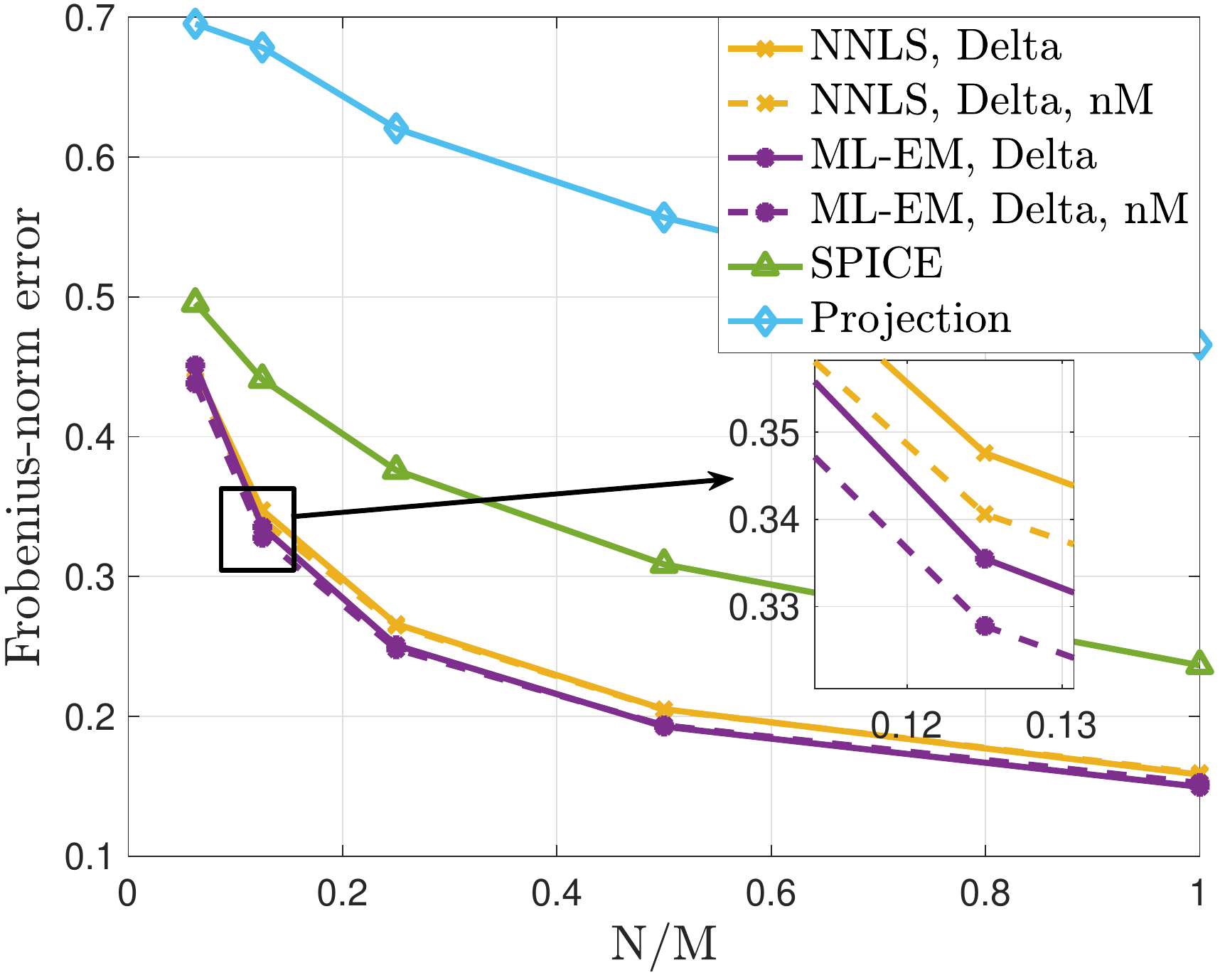}
		\caption{DL Frobenius-norm error}
		\label{fig:NLOS-fro_dl}
	\end{subfigure}
	~
	\begin{subfigure}[b]{\plotwidth\textwidth}
		\includegraphics[width=\textwidth]{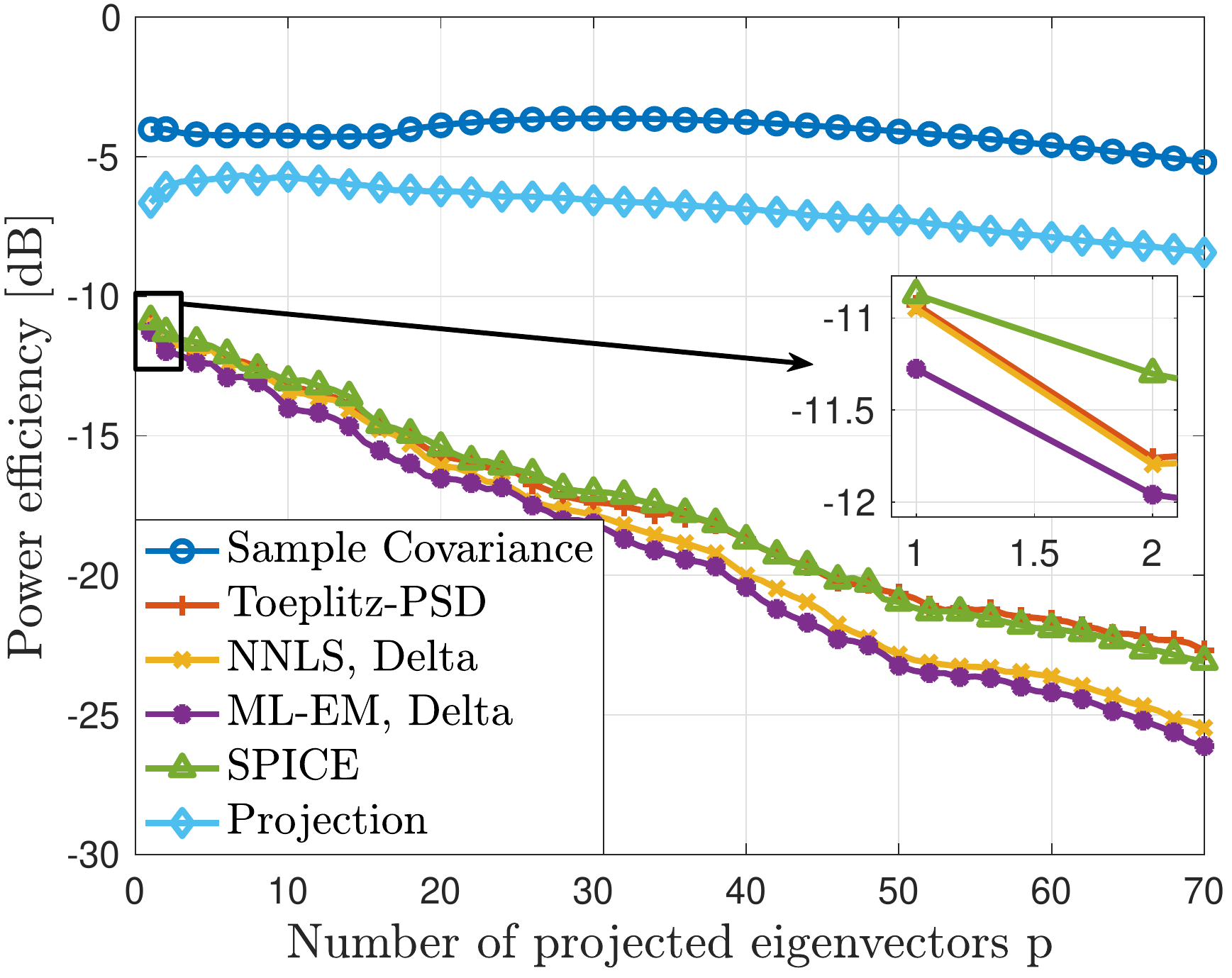}
		\caption{UL PE with $N/M=0.125$}
		\label{fig:NLOS-PE1}
	\end{subfigure}
	~
	\begin{subfigure}[b]{\plotwidth\textwidth}
		\includegraphics[width=\textwidth]{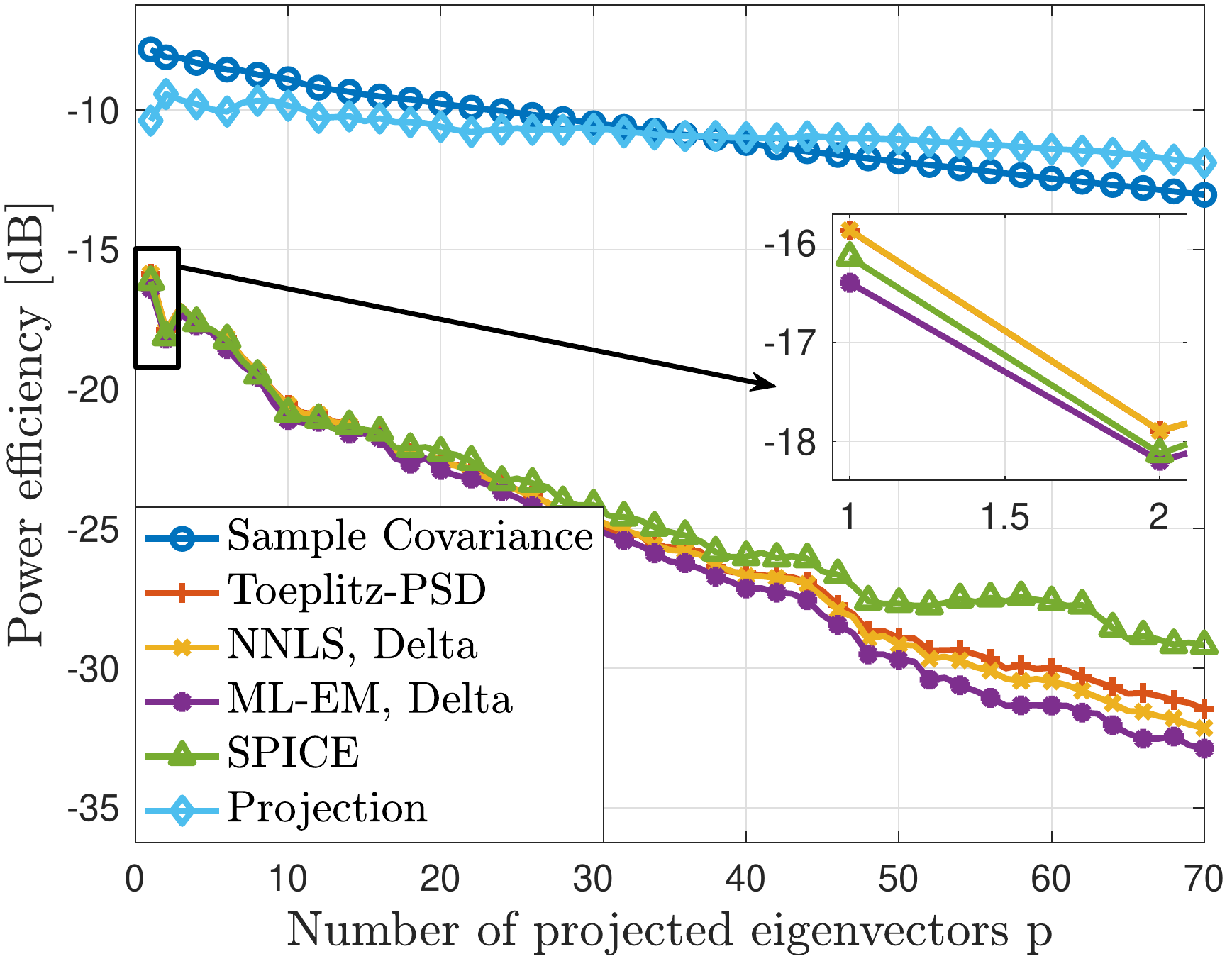}
		\caption{UL PE with $N/M=0.5$}
		\label{fig:NLOS-PE2}
	\end{subfigure}
	~
	\begin{subfigure}[b]{\plotwidth\textwidth}
		\includegraphics[width=\textwidth]{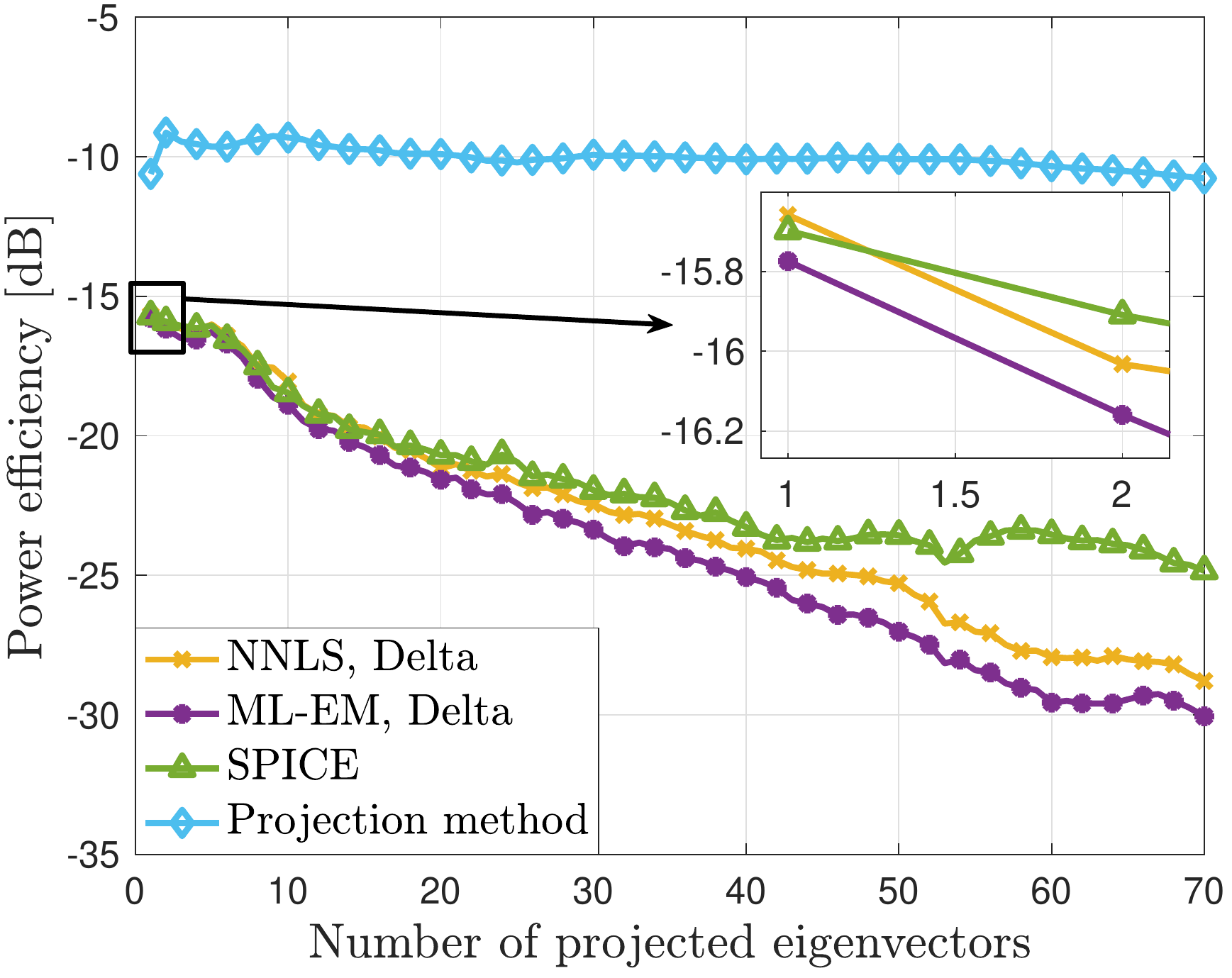}
		\caption{DL PE with $N/M=0.5$}
		\label{fig:NLOS-PE_dl}
	\end{subfigure}
	\caption{ Covariance estimation quality comparison for scenario 3GPP-3D-UMa-NLOS.} \label{fig:NLOS-spike_err}
\end{figure*} 

In Fig.~\ref{fig:spike_err}, the UL and DL channel covariance estimation error for scenario 3GPP-3D-UMa-LOS in terms of the normalized Frobenius-norm error, normalized MSE of channel estimation and power efficiency under different sample ratios $N/M$ (from 0.0625 to 1) are depicted. 
It is observed that the results of the proposed NNLS and ML methods significantly outperform the other benchmarks in both UL and DL for all metrics under all range of sample size. Although the Toeplitz-PSD method has very similar Frobenius-norm errors compared to our methods, it performs significantly worse than our methods in terms of channel estimation error especially under extremely small sample size. Moreover, it is worth to emphasize again that the Toeplitz-PSD method does not give directly an easy way to perform the UL-DL covariance transformation. In contrast (see Fig.~\ref{fig:fro_dl}), the proposed methods yields very good UL-DL transformation (according to the considered metrics) even under very small sample size. 
Furthermore, we also present the results of NNLS and ML-EM without MUSIC (i.e., without explicit spike location estimation) denoted as 'NNLS, Delta, nM' and 'ML-EM, Delta, nM' in Fig.~\ref{fig:fro} and \ref{fig:MSE}. It is observed that in this case the performances degrade dramatically. 
This indicates that the proposed MUSIC step is necessary and non-trivially improves the overall performance. 
It is also noticed that our results without MUSIC are still much better than the results of SPICE and projection method, which shows the advantage of the proposed NNLS and ML-EM algorithms.

Results for the NLOS scenario 3GPP-3D-UMa-NLOS are depicted in Fig.~\ref{fig:NLOS-spike_err}. First, we observe again that our methods outperform the other benchmarks. Interestingly, the results with MUSIC perform almost the same as the results without MUSIC. 
Notice that in this case the scheme is unaware of the fact that there are no spikes, and the MDL/MUSIC step may give some spurious spikes, which are then
eliminated by the subsequent NNLS and ML/EM step (estimated coefficients near zero). 
This demonstrates the fact that the proposed method is robust to both LOS and NLOS cases.

\subsubsection{Comparison with Dirac delta and overlapping Gaussian dictionaries}\label{sec:dictionary}

In this part, we show the results based on Dirac delta and overlapping Gaussian dictionaries to indicate the importance of finding the proper dictionary. We set a fixed $N = 16$, i.e., $N/M = 0.125$. We test the NNLS algorithm with Dirac delta and Gaussian dictionaries as well as the QP estimator with Gaussian dictionaries with $\widetilde{G}=10000$ under different number of dictionaries ($G/M$ is from 0.0625 to 2) for the NLOS scenario 3GPP-3D-UMa-NLOS. The overlapping Gaussian dictionary functions are defined as follows. Given $G$, let $\widetilde{\psi}(\xi)$ be a Gaussian density whose support is limited to $[-\frac{4}{G+3}, \frac{4}{G+3}]$. 
The dictionary $\{\psi_i(\xi): i \in [G]\}$ consists of skewed shifted versions of $\widetilde{\psi}(\xi)$, i.e., $\psi_i(\xi) = J(\xi) \widetilde{\psi}(\xi+1-\frac{2(i+1)}{G+3}),i
\in[G]$, where $J(\xi) = \frac{1}{\sqrt{1-\xi^2}}$ is a factor due to coordinate transformation from $\theta$ to $\xi$. Specifically, $\widetilde{\psi}(\xi)$ is given by
\begin{equation}
    \widetilde{\psi}(\xi) = \begin{cases} 
       \frac{a_0}{\sigma\sqrt{2\pi}} \text{exp}\left(-\frac{(\xi-\mu)^2}{2\sigma^2}\right), &  \; \xi \in [-\frac{4}{G+3}, \frac{4}{G+3}] \\
      0, & \;\text{otherwise}
   \end{cases},
\end{equation}
where $\mu = 0$ and $\sigma = \frac{4}{3(G+3)}$ to ensure that the truncated interval $[-\frac{4}{G+3}, \frac{4}{G+3}]$ accounts for $6\sigma$ of the Gaussian function, and $a_0$ is a normalization scalar such that $\int^1_{-1}\widetilde{\psi}(\xi)d \xi = 1$.

The results are shown in Fig.~\ref{fig:kernel}. From Fig.~\ref{fig:kernel-fro} and \ref{fig:kernel-MSE} it is observed that the results under Gaussian dictionaries are much better than the results under Dirac delta dictionaries when the number of dictionaries is small, e.g., $G/M \leq 1$. It is also observed that both Frobenius-norm error and  channel estimate MSE of NNLS and SPICE under Dirac delta dictionaries decrease dramatically as $G$ increases. In contrast, 
the results under Gaussian dictionaries become slightly worsen when $G/M > 0.5$. 
As $G$ becomes large, the results with Gaussian dictionary functions converge to the Dirac delta case. This is of course explained by the fact that 
the thin Gaussian density with normalized integral are more and more similar to Dirac delta functions. 
A similar behaviour is observed in Fig.~\ref{fig:kernel-PE05} and \ref{fig:kernel-PE15}. These results indicate that by choosing some template dictionary function 
may achieve advantages of the quantization of the AoA domain compared to Dirac delta function for small size of the dictionary, which reduces significantly 
the dimension and computational complexity.
\newcommand{\plotwidthgau}{0.45}
\begin{figure*}[t]
	\centering
	\begin{subfigure}[b]{\plotwidthgau\textwidth}
		\includegraphics[width=\textwidth]{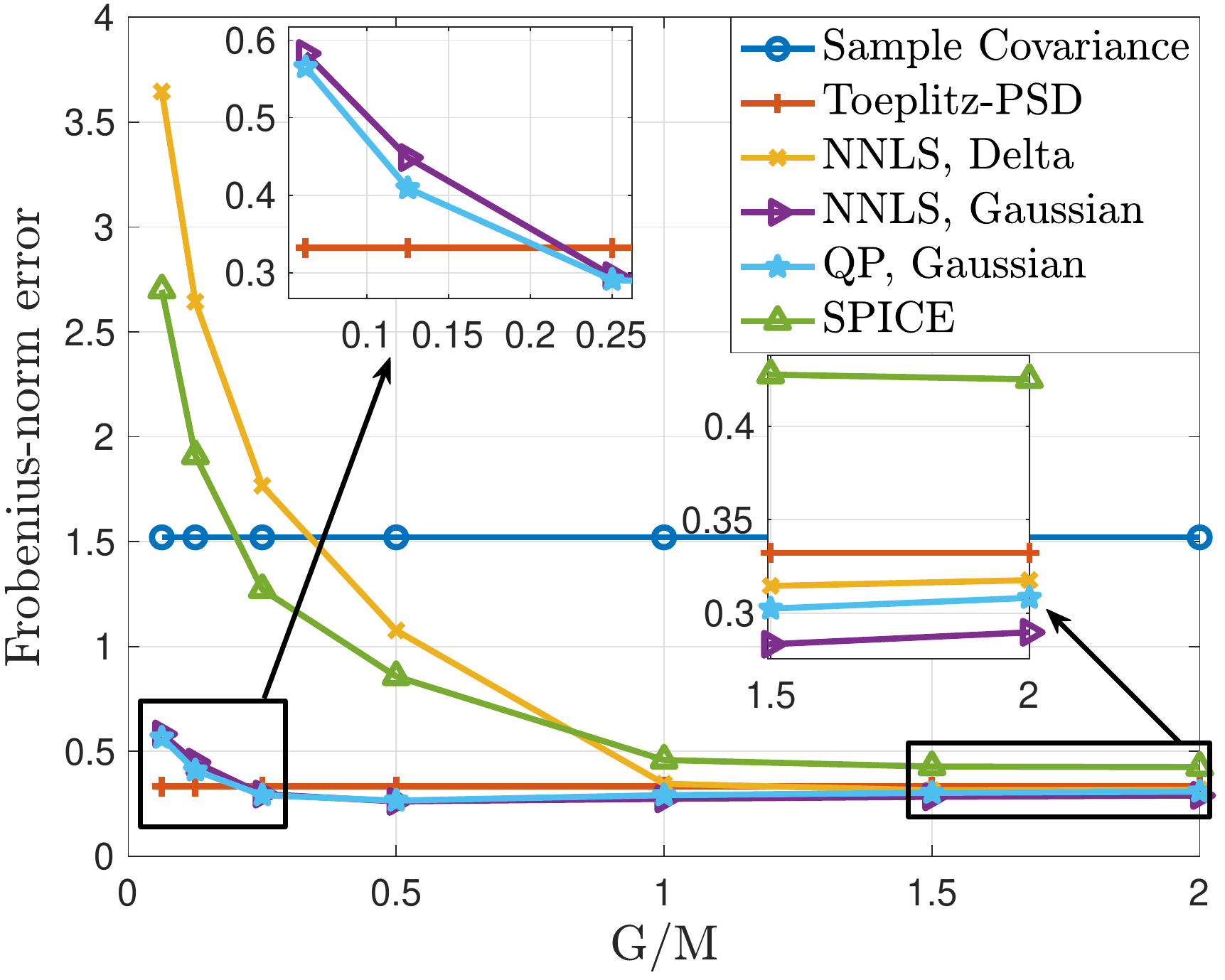}
		\caption{Frobenius-norm error}
		\label{fig:kernel-fro}
	\end{subfigure}
	~ 
	\begin{subfigure}[b]{\plotwidthgau\textwidth}
		\includegraphics[width=\textwidth]{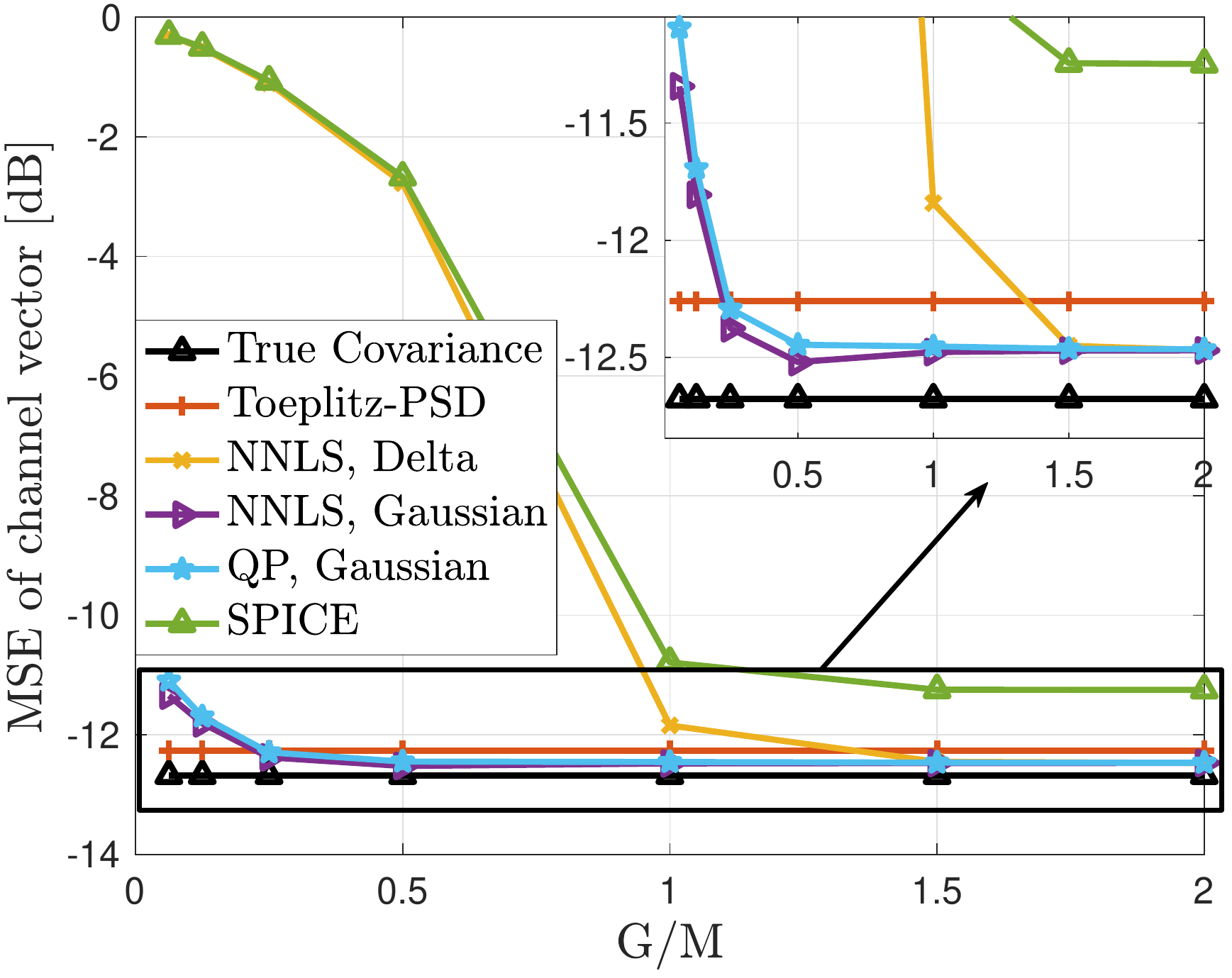}
		\caption{NMSE of channel}
		\label{fig:kernel-MSE}
	\end{subfigure}
	~
	\begin{subfigure}[b]{\plotwidthgau\textwidth}
		\includegraphics[width=\textwidth]{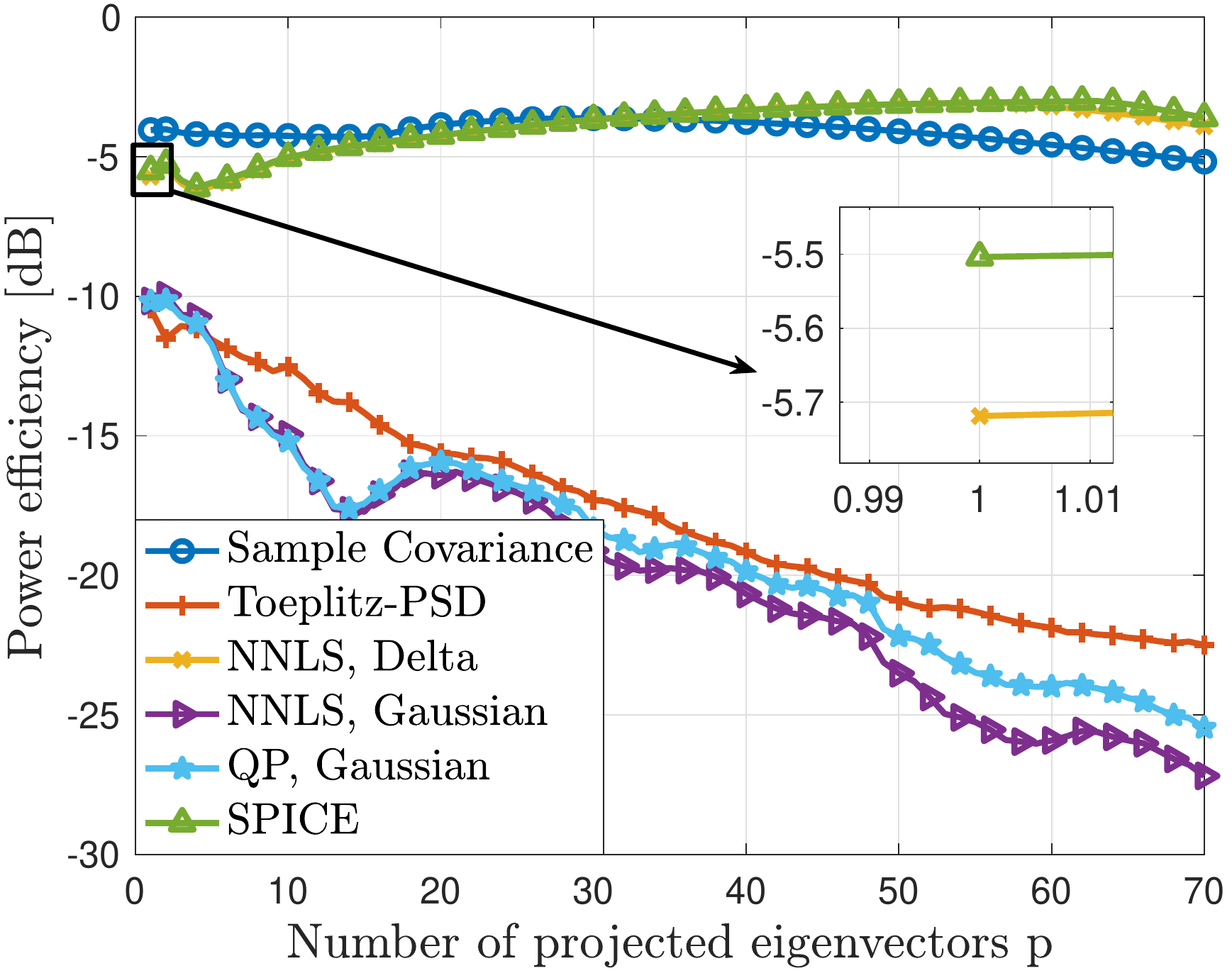}
		\caption{PE with $G/M=0.5$}
	    \label{fig:kernel-PE05}
	\end{subfigure}
	~ 
	\begin{subfigure}[b]{\plotwidthgau\textwidth}
		\includegraphics[width=\textwidth]{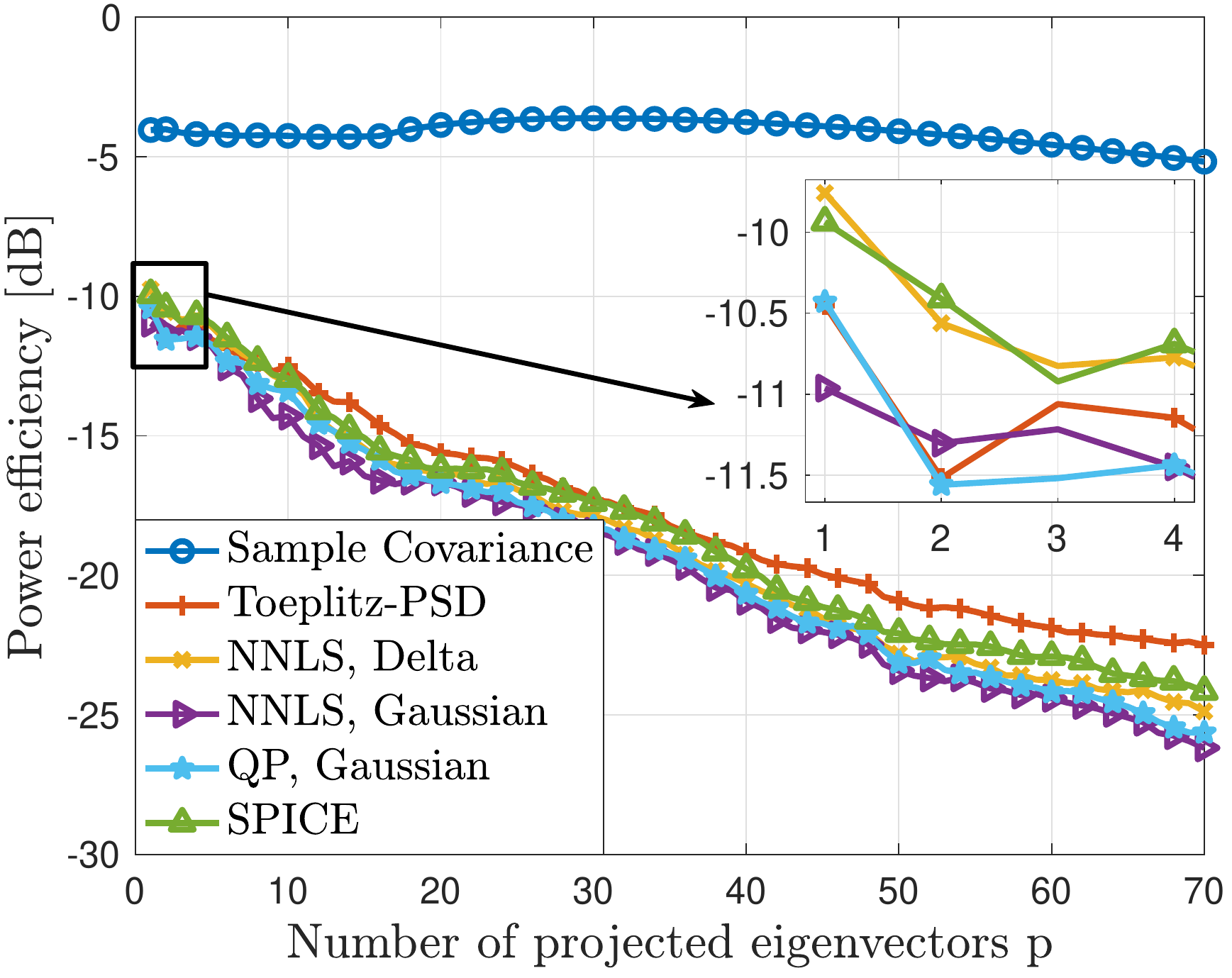}
		\caption{PE with $G/M=1.5$}
		\label{fig:kernel-PE15}
	\end{subfigure}
	
	\caption{Comparison with Dirac delta and Gaussian dictionaries with $N/M=0.125$ under various $G/M$ for scenario 3GPP-3D-UMa-NLOS.} \label{fig:kernel}
\end{figure*}

\section{Conclusion} \label{sec:conclusion}
In this work, we addressed the problem of estimating the covariance matrix of the channel vector from a set of noisy UL pilot observations in massive MIMO systems. By modeling the ASF of the channel as a parametric representation in the angle domain, we proposed the NNLS, QP and ML-EM based estimators to obtain the model parameters.  In order to find the discrete scattering components (number and location of the spikes in the ASF), we adopted MDL and MUSIC methods. Theoretical results guarantee that the separation of the spikes with respect to the clutter of eigenvalues due to the diffuse scattering components is large, for a large number of antennas. This yields that the spikes support estimation in the massive MIMO regime is very reliable. In addition, our method estimates both the coefficients of the spikes and the coefficients of the diffuse part of the ASF jointly. Extensive numerical simulations based on realistic channel emulator QuaDriGa under 3GPP communication scenarios show that the proposed methods are superior to several state-of-the-art algorithms  in the literature in terms of different performance metrics, especially for a small number of samples, which is particularly relevant for the massive MIMO application. 

\appendices 

\section{Consistency of MUSIC}\label{sec:MUSIC_analysis}
We provide an analysis of the consistency of proposed MUSIC algorithm based on the results in \cite{najim2016statistical}. A more challenging scaling regime in  \cite{najim2016statistical} was studied, where the amplitude of the spikes decreases by increasing $M$ such that identifying them becomes more difficult. Specifically, recall in \eqref{eq:gamma_decomp} that the ASF $\gamma(\xi)$ is decomposed into discrete part $\gamma_d(\xi) = \sum^r_{i=1}c_i\delta(\xi-\phi_i)$  and continuous parts $\gamma_c(\xi)$. In the pessimistic scaling regime studied in \cite{najim2016statistical}, the coefficients of spikes $\{c_i^{(M)}: i\in [r]\}$ scale with $M$ according to $c_i^{(M)}=\frac{\upsilon_i}{M}$, where $\{\upsilon_i:i \in [r]\}$ are positive constants encoding the relative strength of the spikes. The result in \cite{najim2016statistical} shows that when the number of antennas $M$ and the sample size $N$ both approach to infinity at the same rate such that $\frac{M}{N} \to \zeta > 0 $, where $\zeta$ is a finite constant, if the separation condition $ \min_{i\in[r]}\upsilon_i > \omega_0(\zeta)$ holds, where $\omega_0(\zeta) \geq \|\gamma_c\|_\infty  +N_0$ is a monotonically increasing function of $\zeta$, then MUSIC is asymptotically consistent, i.e., $M(\widehat{\phi}_i - \phi_i) \xrightarrow[M\to\infty]{\text{a.s.}} 0, \forall i \in [r]$, where a.s. stands for almost surely. In our wireless communication scenario, the spike coefficients $\{c_i: i\in[r]\}$ remain the same independent of the number of BS antennas $M$. We mimic this by assuming that the coefficients $\{\upsilon_i\}$ are also growing proportionally to $M$ like $\upsilon_i = M c_i$. Then, the separation condition would be satisfied for any finite $\zeta$ and for any practically relevant $\gamma_c$ provided that $M$ is sufficiently large. It is worthwhile to emphasize that this result implies that no matter how small the spike amplitudes $\{c_i: i \in [r]\}$ are and no matter how small the number of  samples $N$ is compared with $M$ (of course provided that the asymptotic sampling ratio $\zeta$ remains finite), MUSIC will be able to recover all the spikes if $M$ is sufficiently large. This provides a strong argument in favor of using  MUSIC as the model order and spike location estimation for our application.

\section{Toeplitzation of the Sample Covariance Matrix}\label{sec:teop_appendix}
The Toeplitzed matrix $\widetilde{\Sigmam}_{\hv}$ is obtained by solving the orthogonal projection problem
\begin{equation}\label{eq:toep_projection}
    \widetilde{\Sigmam}_{\hv} = \argmin_{\Sigmam}\; \|\Sigmam - \widehat{\Sigmam}_{\hv} \|^2_{\sfF}, \quad \text{s.t.} \; \Sigmam \text{ is Hermitian Toeplitz}. 
\end{equation}
Define $\sigmav\in\CC^{M\times1}$ as the first column of $\Sigmam$ and  $\mathcal{P}$ as the operation that projects $\boldsymbol{\sigma}$ and its conjugate to a Hermitian Toeplitz matrix as $\Sigmam = \mathcal{P}(\boldsymbol{\sigma})$. We further define 
$
\mathcal{K}_i = \{(r,c): r\leq c \; \text{with} \; [\Sigmam]_{r,c} = [\boldsymbol{\sigma}]_i\}, i\in [M]
$
as the set of all those indices in $\Sigmam$, in which the $i$-th variable $[\boldsymbol{\sigma}]_i$ appears. Then, the objective function in \eqref{eq:toep_projection} can be equivalently presented as
\begin{align}\label{eq:bttb_sigma}
    f(\boldsymbol{\sigma}) &= \underset{\boldsymbol{\sigma}}{\min} \;\; \|\mathcal{P}(\boldsymbol{\sigma}) - \widehat{\Sigmam}_{\hv} \|^2_\sfF=  \underset{\boldsymbol{\sigma}}{\min} \;\; \sum_{i=1}^{M} \sum_{(r,c)\in \mathcal{K}_i}\left|[\boldsymbol{\sigma}]_i - [\widehat{\Sigmam}_{\hv}]_{r,c}\right|^2.
\end{align}
By setting the derivative of $[\boldsymbol{\sigma}]_i$ to zero, the optimum $\sigmav^\star$ is achieved as 
\begin{align}
    \frac{\partial f(\boldsymbol{\sigma})}{\partial[\boldsymbol{\sigma}]_i} &= \sum_{(r,c)\in \mathcal{K}_i} 2([\boldsymbol{\sigma}]_i - [\widehat{\Sigmam}_{\hv}]_{r,c}) \stackrel{!}{=} 0, \quad \forall i \in [M], \\
     \Rightarrow  \quad\quad [\boldsymbol\sigma^{\star}]_i &= \frac{\sum_{(r,c)\in \mathcal{K}_i} [\widehat{\Sigmam}_{\hv}]_{r,c}}{|\mathcal{K}_i|}, \quad \forall i \in [M],
\end{align}
where $|\mathcal{K}_i|$ denotes the number of elements in $\mathcal{K}_i$. Correspondingly, $\widetilde{\Sigmam}_{\hv}=\mathcal{P}(\sigmav^\star)$.

\section{Computational Complexity}  \label{appendix-computational-complexity}
We briefly provide an analysis of the floating-point operations (FLOP) based computational complexity for proposed NNLS and ML-EM algorithms. We adopt the result of FLOP counting for different operations summarized in \cite{hunger2005floating}.\footnote{Although in \cite{hunger2005floating} a FLOP is assumed to be either a complex multiplication or complex summation, we only consider the complex multiplications since in practice the summation operations can be further optimized so that its complexity is much less than the complexity of multiplications.} 
The matrix inverse operation for a positive definite matrix is based on Cholesky decomposition. 

\subsection{Complexity of NNLS algorithm}
In general the algorithms tackling NNLS can be roughly divided into {\em active-set} and {\em projected gradient descent} approaches \cite{chen2010nonnegativit}. The projected gradient descent approaches are known to require a very large number of iterations to converge to the optimum due to very small step size around the optimality to guaranty the convergence \cite{bertsekas2015convex}. As an example, we have numerically tested the sequential coordinate-wise algorithm \cite{franc2005sequential}, which has a very simple update rule but unacceptable large required number of iterations to converge. Therefore, we adopt the standard active-set method \cite{lawson1974solving}, which is based on the observation that only a small subset of constraints are usually active at the solution. Specifically, considering the NNLS problem $\xv^\star = \argmin_{\xv\geq0}\;\|\Am\xv-\yv\|$, the active-set method consists of two nested loops, where we count FLOP of the least squares $[(\Am^P)^\transp\Am^P]^{-1}(\Am^P)^\transp\yv$ in each inner and outer loop as well as the expression $\Am^\transp(\yv-\Am\xv)$ in each outer loop, where $\Am^P$ is a matrix associated with only the variables currently in the active set $P$, please see \cite{lawson1974solving} for more details. The input dimensions of our NNLS problem in \eqref{nnls_toep} are $\Am \in \mathbb{R}^{2M\times\widehat{G}}$ and $\yv\in\mathbb{R}^{2M\times1}$, where $\widehat{G} = G + \widehat{r}$ and $2M$ is due to the complex variables. Assuming the total required number of outer loops is $I$ and inner loops in the $i$-th outer iteration is $J_i$ and ignoring the potential dimension reduction in inner loops, the upper bounded FLOP is given by 
\begin{align}
    \text{FLOP}_{\text{NNLS}} &\leq 4(I+1)\widehat{G}M + \frac{1}{2}\sum_{i=1}^{I}\Big( (1+J_i)\left(6Mi^2 + i^3 + 6Mi + 3i^2\right)\Big), \\ 
    &= 4(I+1)\widehat{G}M + 3M\sum_{i=1}^{I}i + \frac{3}{2}(2M+1)\sum_{i=1}^{I}i^2+ \frac{1}{2}\sum_{i=1}^{I}i^3 +f(J_i),
\end{align}
where $f(J_i) := \frac{1}{2}\sum_{i=1}^{I} J_i\left(6Mi^2 + i^3 + 6Mi + 3i^2\right)$ is the number of FLOP of inner loops. With a averaged number of inner iterations $\widetilde{J}$, the FLOP of NNLS can be approximated as
\begin{align}
    \text{FLOP}_{\widetilde{\text{NNLS}}} \approx  4(I+1)\widehat{G}M + \left(1+\widetilde{J}\right)\left(\frac{1}{8} I^4 + (M+\frac{3}{4})I^3+(3M+\frac{7}{8})I^2 +(2M+\frac{1}{4})I \right).
\end{align}

\subsection{Complexity of ML-EM algorithm}
We count the FLOP of calculating \eqref{eq:EM_mu_sigma} in each E-step and \eqref{eq:M-step} in each M-step. Note that the terms $\Dm^\herm\Ym$ and the gram $\Dm^\herm\Dm$ in \eqref{eq:EM_mu_sigma} only need to be calculated once and used in each iteration. Additionally, $\widehat{\Um}^{(\ell)}$ is a diagonal matrix and its inverse in \eqref{eq:EM_mu_sigma} can be easily obtained by taking the reciprocals of the main diagonal. Assuming the total number of iterations is $I_{\text{EM}}$, the required number of FLOP is given as  
\begin{equation}
    \text{FLOP}_{\text{EM}} = \widehat{G}M\left(N+\frac{1}{2}(\widehat{G}+1)\right) +  I_{\text{EM}} \left(\frac{1}{2}\widehat{G}^3 + \widehat{G}^2\left(N + \frac{3}{2}\right) + \widehat{G}N\right).
\end{equation}

\newpage
	\bibliographystyle{IEEEtran}
	\bibliography{references}

\end{document}